\documentclass[11pt]{article}

\usepackage{hyperref}
\usepackage{enumitem}
\usepackage[ruled,vlined,linesnumbered,noresetcount]{algorithm2e}
\usepackage{amsmath,amssymb,amsthm}

\newtheorem{theorem}{Theorem}[section]
\newtheorem{proposition}[theorem]{Proposition}
\newtheorem{corollary}[theorem]{Corollary}
\newtheorem{lemma}[theorem]{Lemma}
\theoremstyle{definition}
\newtheorem{definition}[theorem]{Definition}

\newtheorem{question}[theorem]{Question}

\newcommand{\pred}{\mathsf{Pred}}
\newcommand{\var}{\mathsf{Var}}
\newcommand{\propvar}{\mathsf{Prop}}
\newcommand{\liff}{\leftrightarrow}

\newcommand{\B}{\Box}
\newcommand{\D}{\Diamond}

\newcommand{\Dh}{\D_h}

\newcommand{\Dv}{\D_v}

\let\oldexists\exists
\renewcommand{\exists}[1][]{{\oldexists_{#1}}}

\newcommand{\ML}{\mathcal{ML}}
\newcommand{\MLu}{\mathcal{ML}_u}

\newcommand{\bfL}{\mathbf{L}}
\newcommand{\K}{\mathbf{K}}

\newcommand{\Kc}{\mathbf{K^\ast}}

\newcommand{\Kfour}{\mathbf{K4}}
\newcommand{\Kfourt}{\mathbf{K4.3}}
\newcommand{\Sfive}{\mathbf{S5}}
\newcommand{\Diff}{\mathbf{Diff}}

\newcommand{\Ku}{\K_u}

\newcommand{\ex}{\mathit{exp}}
\newcommand{\dc}{\mathit{dec}}

\newcommand{\QML}[1][]{\mathcal{Q^{\#}\!ML}_{#1}}
\newcommand{\QMLcf}{\mathcal{QML}}
\newcommand{\QMLun}{\QML[un]}

\newcommand{\QMLone}[1][]{\QML[#1]^{1}}
\newcommand{\QMLonecf}{\QMLcf^{1}}
\newcommand{\QMLoneun}{\QMLone[un]}

\newcommand{\QMLn}[2]{\mathcal{Q^{\#}\!ML}^{#1}_{#2}}

\newcommand{\va}{\mathfrak{h}} 

\newcommand{\quant}[1]{\mathbf{Q}^{\#}\!{#1}}

\newcommand{\QL}{\quant{\bfL}}
\newcommand{\QK}{\quant{\K}}
\newcommand{\QKexp}{\quant{\K}^\ex}
\newcommand{\QKdec}{\quant{\K}^\dc}
\newcommand{\QLexp}{\quant{\bfL}^\ex}
\newcommand{\QLdec}{\quant{\bfL}^\dc}

\newcommand{\QKc}{\quant{\Kc}}

\newcommand{\QKfour}{\quant{\Kfour}}
\newcommand{\QKfourt}{\quant{\Kfourt}}

\newcommand{\Q}{\mathfrak{Q}}
\newcommand{\q}{\textbf{\textit{q}}}
\newcommand{\R}{\mathfrak{R}}
\newcommand{\subR}{\mathfrak{S}}

\newcommand{\tp}[2]{\mathrm{tp}^\M_{#1}[#2]}
\newcommand{\run}[1]{\widehat{#1}}

\newcommand{\Trans}[2]{\mathsf{Trans}({#2})}
\newcommand{\Sym}[1]{\mathsf{Sym}({#1})}
\newcommand{\id}{\mathsf{id}}

\newcommand{\G}{\mathfrak{G}}
\newcommand{\M}{\mathfrak{M}}
\newcommand{\F}{\mathfrak{F}}
\newcommand{\V}{\mathfrak{V}}

\newcommand{\I}{\mathfrak{I}}

\renewcommand{\phi}{\varphi}
\newcommand{\ie}{\emph{i.e.}}
\newcommand{\fmp}{fmp}
\renewcommand{\iff}[1][]{\stackrel{{\normalfont\mbox{\scriptsize #1}}}{\Longleftrightarrow}}

\newcommand{\dom}{\mathfrak{d}}

\newcommand{\Pow}[1]{{2^{#1}}}
\newcommand{\size}[1]{\left\Vert{#1}\right\Vert}
\newcommand{\card}[1]{{\left|{#1}\right|}}

\newcommand{\sub}[1]{\mathsf{sub}(#1)}
\newcommand{\md}[1]{\mathsf{md}({#1})}
\newcommand{\cpt}[1]{\mathsf{cap}(#1)}

\newcommand{\etimes}{\times^{\ex}}
\newcommand{\dtimes}{\times^{\dc}}

\newcommand{\PSpace}{\textsc{PSpace}}

\newcommand{\ExpTime}[1][]{\textrm{#1}\textsc{ExpTime}}
\newcommand{\NExpTime}[1][]{\textsc{N}\textrm{#1}\textsc{ExpTime}}

\newcommand{\lefttoright}{\item[($\Rightarrow$)]}
\newcommand{\righttoleft}{\item[($\Leftarrow$)]}

\newcommand{\QKworld}[1][]{\textsf{QK-WORLD}_{#1}}
\newcommand{\Kworld}{\textsf{K-WORLD}}

\usepackage{todonotes}
\newcommand{\nb}[1]{\textcolor{red}{$\blacksquare$}\todo{\footnotesize #1}}

\usepackage{environ}
\NewEnviron{killcontents}{}

\newenvironment{draftproof}
	{\color{black!40}\killcontents}
	{\color{black}\endkillcontents}
	
\usepackage{authblk}

\usepackage{scrtime}
\begin{document}
\title{Decidable fragments of first-order modal logics with counting quantifiers over varying domains}
\author{Christopher Hampson\\ \small \vspace{-5pt}Department of Informatics, King's College London\\ \vspace{-10pt}}

\maketitle
\begin{abstract}
This paper explores the computational complexity of various natural one-variable fragments of first-order modal logics with the addition of counting quantifiers, over both constant and varying domains. 
The addition of counting quantifiers provides us a rich language with which to succinctly express statements about the quantity of objects satisfying a given first-order property, using a single variable. 
Optimal \NExpTime\ upper-bounds are provided for the satisfiability problems of the one-variable fragment of the minimal first-order modal logic $\mathbf{QK}$, over both constant and expanding/decreasing domain models, where counting quantifiers are encoded as binary strings. 
For the case where the counting quantifiers are encoded as unary strings, or are restricted to a finite set of quantifiers, 
it is shown that the satisfiability problem over expanding domains is \PSpace-complete, whereas over decreasing domains the problem is shown to be \ExpTime-hard.

%

\end{abstract}

\section{Introduction}
\label{sec:intro}

\newcommand{\cell}[2][c]{\begin{tabular}{#1}#2\end{tabular}}

Following the negative results of Church and Turing on the \emph{Entschiedungsproblem}, there was a surge of interest in the 1920--30s in establishing where lay the boundary between decidable and undecidable fragments of classical first-order logic (see the monograph~\cite{Boerger2001}, and references therein). First-order \emph{modal} logics, by contrast, have enjoyed far less attention, despite their numerous applications in temporal query languages~\cite{Chomicki1994,Chomicki1995} and in modelling of distributed or multi-agent systems~\cite{Fagin2004,Belardinelli2009}.

This lack of attention is arguably not a result of disinterest, but rather a result of the relatively weak expressive power required to obtain undecidability results. 
For example, while the two-variable fragment of classical first-order logic is decidable~\cite{Gradel1997}, even the two-variable, monadic fragment
of many first-order modal logics is already undecidable~\cite{Kripke1962}. 
There is, therefore, a practical demand for finding expressive fragments of first-order modal logics that retain their decidability, and establishing where the boundary between decidability and undecidability lies.


%
One approach is to consider the one-variable fragment, which is typically decidable whenever the underlying modal logic is decidable~\cite{missing}. However, these are often too weakly expressive for most practical uses. 
More expressive, yet still often decidable, are the \emph{monodic} fragments, in which the modalities \emph{de re} (about the thing) are restricted to formulas containing at most one free variable, with no restrictions are placed on modalities \emph{de dicto} (about the statement)~\cite{HodkinsonEtal2000,Wolter2001}.

Another approach is to expand the language with countably many \emph{counting quantifiers} $\exists[\leq m]$ able to succinctly express statements about the quantity of objects satisfying a given property, without requiring many auxiliary variables to address each such object independently. 
Our choice of how to encode the numerical bound $m$ associated with our counting quantifiers may potentially affect the computational complexity of our satisfiability problem, depending on whether we choose a succinct \emph{binary} encoding for $m$ or a more verbose \emph{unary} encoding. For the purposes of upper bounds, it is the binary encoding that is the more stringent.
Counting quantifiers, in the form of cardinality restrictions on concepts, also play a central role in several expressive fragments of description logics~\cite{BaaderEtal1996,Baader2003}.

It is well-known that counting quantifiers can be safely added to the two-variable fragment of classical first-order logic without increasing the complexity of its satisfiability problem. Moreover, it makes no difference, whether our counting quantifiers are encoded as unary strings~\cite{Pacholski1997} or as binary strings~\cite{PrattHartmann2005}. 
They, therefore, suggest an attractive direction in the quest to gain greater expressive power from finite variable fragments of first-order modal logics, without jeopardizing their decidability. Some examples of first-order formulas with counting quantifiers include:
\begin{itemize}
	\item \emph{``At most four components are believed to be faulty''}:\\ $\exists[\leq 4]x\big(\mathsf{Component}(x) \land \B\mathsf{Faulty}(x)\big)$
	\item \emph{``It is possible that there is life on more than one planet''}:\\ $\D\exists[\geq 2]x\big(\mathsf{Planet}(x) \land \mathsf{hasLife}(x)\big)$, 
	\item Generalised Barcan formula:\\ $\exists[\leq c]x \D P(x) \to \D\exists[\leq c]x\, P(x)$, for $c<\omega$.
\end{itemize}
Unfortunately, the the addition of counting quantifiers is not always so amicable in the case of first-order modal logics, where even the most modest admission of counting quantifiers can result in a jump from decidability to the undecidable~\cite{Hampson2012}, or even the highly undecidable~\cite{Hampson2015}.

Turning our attention towards more practical first-order modal logics, this paper explores several decidable one-variable fragments over both constant and varying domains that are permitted to either expand or contract relative to the modal accessibility relation. Furthermore, we shall see that over varying domains, the choice of encoding for the numerical bounds appearing in the counting quantifiers can have a dramatic effect on the overall complexity. As too can placing any finite bound on subscripts appearing in the quantifiers, regardless of their encoding. Table~\ref{tab:results}, below, summarizes the results contained herein.

\begin{table}[ht!]\footnotesize
\centerline{
\begin{tabular}{|c||c|c|c|}
\hline & \cell{Constant\\ domains} & \cell{Expanding\\ domains} & \cell{Decreasing\\ domains}\\
\hline\hline
	\cell{Binary\\ Quantifiers\\ (Unbounded)} 
			& \cell{\NExpTime-complete\\[2pt] Theorem~\ref{thm:QK-unbounded}}
			& \cell{\NExpTime-complete\\[2pt] Theorem~\ref{thm:QKexp-unbounded}}
			& \cell{\NExpTime-complete\\[2pt] Theorem~\ref{thm:QKdec-unbounded}}
			\\
\hline
	\cell{Binary\\ Quantifiers\\ (Bounded)}
			& \cell{\NExpTime-complete\\[2pt] Theorem~\ref{thm:QK-unbounded}}
			& \cell{\PSpace-complete\\[2pt] Theorem~\ref{thm:QKexp-bounded}}
			& \cell{\ExpTime-hard,\\ in \NExpTime\quad  \fbox{?}\\[2pt] Theorem~\ref{thm:QKdec-bounded}} 
			\\
\hline
	\cell{Unary\\ Quantifiers}
			& \cell{\NExpTime-complete\\[2pt] Theorem~\ref{thm:QK-unbounded}}
			& \cell{\PSpace-complete\\[2pt] Theorem~\ref{thm:QKexp-bounded}}
			& \cell{\ExpTime-hard,\\ in \NExpTime\quad  \fbox{?}\\[2pt] Theorem~\ref{thm:QKdec-bounded}} 
			\\
\hline
\end{tabular}}

\caption{Summary of results contained herein}
\label{tab:results}
\end{table}

\subsection{Outline of paper}

In Section~\ref{sec:definitions}, we introduce the the one-variable fragment of first-order logic with counting quantifiers and establish the definitions for the fragments of first-order modal logics with which we shall be working.  
In Section~\ref{sec:constant} we consider first-order logics over over constant domains and show that the satisfiability problem for the one variable fragment of $\QK$ where our quantifier subscripts are encoded as binary strings, is \NExpTime-complete. 
This results is extended, in Section~\ref{sec:unbounded}, to logics over expanding and decreasing domains.

In Section~\ref{sec:expanding}, we investigate the affect that the choice of encoding of the counting quantifiers has on the complexity, over expanding domains, while Section~\ref{sec:decreasing} presents partial results over decreasing domains. 
%
%
In Section~\ref{sec:bimodal} we explore the connections between the aforementioned first-order modal logics with counting quantifiers and two-dimensional products of propositional modal logics related to von~Wright's logic of `elsewhere'.
We conclude with a discussion of some open problems in Section~\ref{sec:discussion}.

\section{Preliminaries}
\label{sec:definitions}
\newcommand{\bbN}{\mathbb{N}}
\subsection{First-order Modal Logics with Counting Quantifiers}

The reader is assumed to be familiar with the basics of propositional and first-order modal logics, such as can be found in Fitting and Mendelsohn~\cite{FittingMendelsohn1998}, for example. 
In what follows, we shall consider the first-order modal language comprising a countably infinite set of predicate symbols $\pred=\{P_0,P_1,\dots\}$, each with an associated arity, together with a set of first-order variables $\var=\{x_0,x_1,\dots\}$. We will write $\QMLcf$ to denote the set of all traditional (counting-free) first-order modal formulas with a sole quantifier $\exists$, and distinguish this from the set $\QML$ of all first-order modal formulas with \emph{counting quantifiers} defined by the following grammar:
\begin{equation*}
	\phi\ ::=\ P_i(x_{i_1},\dots, x_{i_n}) \ \mid \ \neg \phi \ \mid \ (\phi_1\land \phi_2) \ \mid \ \D\phi \ \mid \ \exists[\leq c]x\ \phi
\end{equation*}
where $P_i\in \pred$ is an $n$-ary predicate symbol, $x,x_{i_1},\dots, x_{i_n}\in \var$ are first-order variables, and $c\in \bbN$ is encoded as a \emph{binary string}. For the language where counting quantifier subscripts are encoded as \emph{unary} strings, we denote $\QMLun$. 
The other Boolean connectives are defined in the usual way, with the addition of the usual dual modal operator $\B \phi := \neg \D \neg \phi$. 
We will identify the counting free fragment $\QMLcf$ to be a sublogic of $\QML$ with the abbreviations $\exists x \phi:=\neg \exists[\leq 0] x \phi$ and $\forall x \phi := \neg\exists x \neg \phi = \exists[\leq 0] x \neg\phi$.

%

Other counting quantifiers $\exists[\geq c] x \phi$ and $\exists[=c] x \phi$ can similarly be expressed in terms of $\exists[\leq c] x \phi$. It should be noted that the `obvious' definition $\exists[=c]x \phi :=\exists[\leq c]x \phi \land \exists[\geq c]x \phi$ leads to an exponential increase in the size of the formula over the succinct abbreviation. This, however, can be avoided by rewriting $\exists[=c] \phi := \exists[\leq c]x Q(x) \land \exists[\geq c]x Q(x) \land \forall x (Q(x) \liff \phi)$, for some fresh monadic predicate symbol $Q\in \pred$.

We distinguish between the fragments $\QML$ having \emph{unbounded} counting quantifiers and the fragments $\QMLn{}{k}$ whose formulas that do not contain quantifiers with subscripts larger than $k<\omega$. 
For each $\ell<\omega$, let $\QMLn{\ell}{}$ denote the $\ell$-variable fragment comprising only those formulas containing the variables $x_1, \dots, x_\ell$,
and denote by $\QMLn{\ell}{k}=\QMLn{\ell}{}\cap \QMLn{}{k}$ the $\ell$-variable fragment with quantifiers subscripts not exceeding $k$.\\

We define $\sub{\phi}\subseteq \QML$ to be the set of all \emph{subformulas} of $\phi$, 
$\md{\phi}<\omega$ to be the \emph{modal depth} of $\phi$, taken to be the maximum nesting depth of modal operators, 
and $\cpt{\phi}<\omega$ to be the \emph{capacity} of $\phi$, taken to be the value of the largest quantifier subscript occurring in~$\phi$.

\subsection{Semantics}

Formulas of $\QML$ are interpreted in \emph{first-order Kripke models} of the form $\M=(\F,D,\dom,\I)$, where $\F=(W,R)$ is a unimodal \emph{Kripke frame} comprising a set of possible wolds $W$ and a binary accessibility relation $R\subseteq W\times W$ on $W$, $D$ is a non-empty set of \emph{domain objects} from which the domain function $\dom:W\to (\Pow{D}-\{\emptyset\})$ selects a non-empty subset $\dom(w)\subseteq D$, for each $w\in W$. Finally, $\I:W\times \pred \to \bigcup_{n<\omega} D^n$ is a function associating each $w\in W$ and and each $n$-ary predicate symbol $P_i\in \pred$ with a $n$-ary relation $\I(w,P_i)\subseteq \dom(w)^n$ on $\dom(w)$. 

In general, there may be no proscription on the behaviour of the domain function $\dom$, in what we call \emph{varying domain models}. However, of most interest to us are those cases where the domain of interpretation is constant or is permitted to either only expand or only contract relative to the direction of the underlying Kripke frame. 
We say that the model is:
\begin{itemize}
	\item[--] a \emph{constant domain model} if $\dom(u)=\dom(v)$ for all $u,v\in W$,
	\item[--] an \emph{expanding domain model} if $\dom(u)\subseteq \dom(v)$, whenever $u R v$, and 
	\item[--] a \emph{decreasing domain model} if $\dom(u)\supseteq \dom(v)$, whenever $u R v$.
\end{itemize}

\smallskip
Given a model $\M=(\F,D,\dom,\I)$ and a variable assignment $\va:\var\to D$, mapping variables to domain objects, we define satisfiability in $\M$ in the standard way by taking:
\begin{equation*}
	\begin{array}{lcl}
	\M,w\models^\va P_i(x_1,\dots, x_n) &\ \iff \ &\big(\va(x_1),\dots, \va(x_n)\big)\in \I(w,P_i),\\[5pt]
	\M,w\models^\va \neg \phi & \iff & \M,w\not \models^\va \phi,\\[5pt]
	\M,w\models^\va (\phi_1\land \phi_2) & \iff & \M,w\models^\va \phi_1 \mbox{ and } \M,w\models^\va \phi_2,\\[5pt]
	\M,w\models^\va \D\phi & \iff & w R v \mbox{ and } \M,v\models^\va\phi,\ \mbox{for some $v\in W$},
	\end{array}
\end{equation*}
for all $w\in W$, where $P_i\in \pred$ is an $n$-ary predicate symbol. Counting quantifiers are interpreted, by taking:
\begin{equation*}
	\M,w\models^\va (\exists[\leq c] x \ \phi) \quad  \iff \quad \card{\{a\in D\ :\ \M,w\models^{\va(x/a)} \phi\}}\ \leq \ c,
\end{equation*}
for $c\in \bbN$, where $\card{X}$ denotes the cardinality of $X$, and $\va(x/a):\var \to D$ is the variable assignment that agrees with $\va$ on all variables except $x$, for which it assigns the value $a\in D$. 
In the case where $\phi$ contains only a single variable $x\in \var$, it will be convenient to write $\M,w\models^a \phi$ in place of $\M,w\models^\va \phi$, where $a=\va(x)$ is the assignment of $x$. 

%

%

We say that a formula $\phi$ is valid in $\M$ if $\M,w\not\models^\va \neg \phi$, for all $w\in W$ under any variable assignment $\va$. 
For each propositional modal logic $\bfL$, there is an associated first-order modal logic $\QL$ (resp. $\QLexp$, $\QLdec$) taken to be the set of all first-order modal formulas that are valid in every constant (resp. expanding, decreasing) domain model $\M=(\F,D,\dom,\I)$ whose underlying Kripke frame $\F$ is a frame for $\bfL$. 
We say that $\phi$ is satisfiable with respect to $\QL$ (resp. $\QLexp$, $\QLdec$) if $\phi$ is satisfiable in some constant (resp. expanding, decreasing) domain model whose underlying Kripke frame is a frame for $\bfL$. 

The fragment $\QL\cap \mathcal{K}$ is said to have the poly-size (resp. exponential-sized) domain property if every formula $\phi\in \mathcal{K}$ that is satisfiable with respect to $\QL$ can be satisfied in a model in which $D$ is at most polynomial (resp. exponential) in the size of $\phi$. 
Similarly, we say that $\QL\cap \mathcal{K}$ poly-size (resp. exponential-sized) model property if every formula $\phi\in \mathcal{K}$ that is satisfiable with respect to $\QL$ can be satisfied in a model in which both $\F$ and $D$ are at most polynomial (resp. exponential) in the size of $\phi$.

\newcommand{\comment}[1]{}
\comment{
\begin{definition}[Finite domain/model properties]
	Given\nb{revise} a monotonic function $f:\omega\to\omega$, 
	we say that a fragment $\mathcal{K}\subseteq \QML$ has the \emph{$f$-size domain property} with respect to a first-order modal logic $\QL$ if every formula $\phi$ that is satisfiable with respect to $\QL$ can be satisfied in a model $\M=(\F,D,\dom,\I)$ for $L$, with $\F=(W,R)$, where $\card{\dom(w)}\leq f(\size{\phi})$, for all $w\in W$. 
Furthermore, we say that $\mathcal{K}$ has that \emph{$f$-size finite model property (fmp)} with respect to $L$ if $\card{W}\leq f(\size{\phi})$.
	Furthermore, we say that $L$ has the \emph{$f$-size finite model property (fmp)} if we also have that $\card{W}\leq f(\size{\phi})$. 

We say that $\mathcal{K}$ has the \emph{poly-size domain property} (resp.\emph{poly-size fmp}) with respect to $L$ if $f(n)\in O(n^k)$ can be chosen to be a polynomial function, and the \emph{exponential domain property} (resp. \emph{exponential fmp}) with respect to $L$ if $f(n)\in O(2^{n^k})$ can be chosen to be an exponential function, for some constant~$k$	
\end{definition}}

Throughout the remainder of this paper, we will be chiefly concerned with various fragments of the logics $\QK$, $\QKexp$ and $\QKdec$, characterized by the class of all frames. However, owing to a standard `bulldozing' argument (see~\cite{BdRV01}), we may assume without any loss of generality that if $\phi$ is satisfiable with respect to $\QK$ (resp. $\QKexp$, $\QKdec$), then it must be satisfiable in a model whose underlying Kripke frame is an irreflexive, intransitive tree of depth $\leq \md{\phi}$.

\section{Logics with unbounded quantifiers over constant domains}
\label{sec:constant}

We note, first, that the satisfiability problem for the counting-free fragment of $\QK$ in one-variable is \NExpTime-complete~\cite{Marx1999}, while that of the two-variable monadic fragment is already undecidable~\cite{Kripke1962,GabbayShehtman1993}. 
The effect of adding even the most modest counting quantifiers $\exists[\leq c] x$, for $c\in \{0,1\}$, to the decidable one-variable fragment of $\QKfourt$, whose models are based on \emph{linear} frames, is known to result in undecidability~\cite{Hampson2015}\footnote{Many of these results are presented in the context of two-dimensional propositional modal logics, the connection with which is explored in Section~\ref{sec:bimodal}.}.

In marked contrast to these negative results, in this section we show that endowing the one-variable fragment $\QK\cap \QMLonecf$ with infinitely many counting quantifiers $\exists[\leq m]$, for each $m\in \bbN$, does not result in any increase in the computational complexity of its satisfiability problem over its counting-free counterpart; indeed the problem remains \NExpTime-complete. This result was first proved in~\cite{Hampson2016}.

\begin{theorem}
\label{thm:QK-unbounded}
	The satisfiability problem for the one-variable fragment $\QK\cap\QMLone$ with unbounded counting quantifiers is \NExpTime-complete.
\end{theorem}

We show that for a formula $\phi$ to be satisfiable with respect to $\QK\cap \QMLone$, it must be satisfiable in a model that is at most exponential in the size of $\phi$; \ie the one-variable fragment $\QK\cap\QMLone$ has the exponential \fmp. 
This provides us with an effective mechanism by which to check the satisfiability of any formula $\phi\in \QMLone$ --- nondeterministically select a `small' model whose size is at most exponential in $\phi$, and check whether it satisfies $\phi$. Since model-checking can be performed in polynomial time in the size of the model and size of the formula, we obtain an \NExpTime\ upper-bound on the complexity of the satisfiability problem for $\QK\cap \QMLone$.

\begin{theorem}\label{thm:QKfmp}
	The one-variable fragment $\QK\cap \QMLone$ with unbounded quantifiers has the exponential \fmp.
\end{theorem}

To prove this, we employ a version of the method of \emph{quasimodels}~\cite{Wolter2000,GKWZ03}.
Our quasimodels closely resemble full Kripke models, however, each first-order structure is replaced with a \emph{quasistate}, which can be finitely represented. The basic structure of our quasimodels may still be infinite, so we require additional non-trivial `pruning' techniques to ensure that large quasimodels can be reduced to smaller finite quasimodels without sacrificing satisfiability.
Therein lies the crux of the problem we must~solve.\\

First, let us fix some arbitrary first-order modal formula $\phi\in \QMLone$, and throughout what follows, let $n=\card{\sub{\phi}}$ denote the number of subformulas of~$\phi$, $m=\md{\phi}$ denote the modal depth of $\phi$, and $C=\cpt{\phi}$ denote the value of the largest quantifier subscript occurring in $\phi$. In particular we note that $n,m\leq \size{\phi}$, while $C\leq 2^{\size{\phi}}$, owing to the binary encoding of the quantifier subscripts.

\medskip

\begin{definition}[Types and Quasistates]
We define a \emph{type} for $\phi$ to be any subset $t\subseteq \sub{\phi}$ that is \emph{Boolean-saturated} in the sense that:
\begin{enumerate}\itemindent=1.0em\itemsep=6pt
	\item[{\bf (tp1)}] for all $\neg\psi\in \sub{\phi}$, $\neg\psi \in t$ if and only if $\psi\not\in t$, and
	\item[{\bf (tp2)}] for all $(\psi_1 \land \psi_2) \in \sub{\phi}$, $(\psi_1\land \psi_2) \in t$ if and only if $\psi_1 \in t$ and~\mbox{$\psi_2\in t$}.
\end{enumerate}
A \emph{quasistate} for $\phi$ is defined to be a pair $(T,\mu)$ such that: 
\begin{enumerate}[label={\bf (qs\arabic*)}]\itemindent=1.0em\itemsep=6pt
	\item $T$ is a non-empty set of types for $\phi$,
	\item $\mu:T \to \{1,\dots, C,C+1\}$ is a bounded \emph{`multiplicity'} function,
	\item \label{qs:saturation}\emph{($\exists[\leq c]$-saturation)} For all $t\in T$ and $(\exists[\leq c] x\ \xi)\in \sub{\phi}$,
	\begin{equation*}
		(\exists[\leq c] x\ \xi) \in t \qquad \iff \qquad \sum_{t'\in T(\xi)} \mu(t') \ \leq \ c,
	\end{equation*}
	where $T(\xi) = \{t\in T\ :\ \xi\in t\}$ denotes the set of types belonging to $T$ that contain the formula $\xi\in \QMLone$.
\end{enumerate}
\end{definition}

\smallskip

Note that the size of each quasistate cannot exceed the number of distinct types for $\phi$, which is to say that $\card{T}\leq 2^n$, and $\card{\mu}\leq (C+1)\cdot\card{T}\leq (C+1)2^n$, since $\mu\subseteq T\times \{1,\dots, C+1\}$.
The multiplicity function indicates how many `duplicates' of each type are required in order to transform the quasistate into an appropriate first-order structure. 
Note that $\phi$ is indifferent to any duplicates in excess of the value of its largest quantifier subscript, and does not discern between `large' quantities which are beyond its `vocabulary'.

\begin{definition}[Quasimodels]
A \emph{basic structure}\index{basic structure} for $\phi$ is a triple \mbox{$(W,\prec,\q)$}, where $(W,\prec)$ is an intransitive, irreflexive tree of depth $\leq m$, and $\q$ is a function associating each $w\in W$ with a quasistate $\q(w) = (T_w,\mu_w)$. 
A \emph{run} through $(W,\prec,\q)$ is a function $r:W \to \bigcup_{w\in W}T_w$ associating each $w\in W$ with a type~$r(w)\in T_w$. 
A \emph{quasimodel for $\phi$} is a 5-tuple \mbox{$\Q=(W,\prec,\q,I,\R)$} such that:
\begin{enumerate}[label={\bf (qm\arabic*)}]\itemindent=1.5em\itemsep=6pt
	\item \label{qm:basic} $(W,\prec,\q)$ is a basic structure for $\phi$, $I$ is a non-empty set of \emph{indices}, and \mbox{$\R=\{r_i : i\in I\}$} is an set of runs through $(W,\prec,\q)$ indexed by $I$,
	
	\item \label{qm:witness} There is some $w_0\in W$ and $t_0\in T_{w_0}$ such that $\phi\in t_0$, 
	
	\item \label{qm:coherence}\textit{(coherence)} For all $i\in I$, $w\in W$ and $\D\xi \in \sub{\phi}$,
\begin{equation*}
\exists v\in W; \ w \prec v \mbox{ and } \xi \in r_i(v) \quad \Longrightarrow \quad  \D\xi\in r_i(w),
\end{equation*}

	\item\label{qm:saturation}\textit{(saturation)} For all $i\in I$, $w\in W$ and $\D\xi \in \sub{\phi}$,
\begin{equation*}
 \D\xi\in r_i(w)\quad \Longrightarrow \quad \exists v\in W; \ w \prec v \mbox{ and } \xi \in r_i(v),
\end{equation*}

	\item \label{qm:runs}For all $w\in W$ and $t\in T_w$,
	\begin{equation*}
		\mu_w(t) = \min\big(\card{\{i\in I: r_i(w)=t\}},\ C+1\big).
	\end{equation*}
	
\end{enumerate}
\end{definition}

\medskip

The following lemma establishes that our quasimodels precisely capture the notion of satisfiability with respect to $\QK$, and that every quasimodel for $\phi$ can be effectively transformed into model for $\phi$ of proportional size.

\medskip

\begin{lemma}
\label{lem:qm}
Let $\phi\in \QMLone$ be an arbitrary formula in one-variable. 
Then $\phi$ is satisfiable with respect to $\QK$ iff there is a quasimodel for $\phi$.
\end{lemma}
\begin{proof}
Suppose that $\phi$ is satisfiable with respect to $\QK$. 
Then $\M,w_0\models^{a_0} \phi$ for some first-order Kripke model $\M=(\F,D,\dom,\I)$, where $\F=(W,R)$ is an irreflexive, intransitive tree of depth $\leq m$, with $w_0\in W$.

With each $w\in W$ and $a\in D$, we associate the type 
\begin{equation*}
	\tp{w}{a} = \{\xi \in \sub{\phi}\ :\ \M,w \models^a \xi\},
\end{equation*}
and define a basic structure $(W,R,\q)$, by taking $\q(w)=(T_w,\mu_w)$, for all $w\in W$, where
\begin{eqnarray*}
	T_w & = & \{\tp{w}{a}\ :\ a\in D\},\quad \mbox{and}\\
	\mu_w(t) & = & \min\left(\card{\{a\in D : \tp{w}{a}=t\}},C+1\right),
\end{eqnarray*}
for all $t\in T_w$. 
It is straightforward to check that $\q(w)$ is a quasistate, for each $w\in W$. %
Indeed, suppose that $\tp{w}{a}\in T_w$ and that $(\exists[\leq c] x\ \xi)\in \sub{\phi}$, for some $c\leq C$. Then we have that: 
	{\small 
\begin{align*}
	(\exists[\leq c] x\ \xi)\in \tp{w}{a} \quad & \iff \quad \M,w\models^{a} (\exists[\leq c] x\ \xi) \quad \mbox{by definition},\\
		& \iff \quad  \card{\{b\in D\ :\ \M,w\models^{b} \xi\}}\ \leq\ c,\\
		& \iff \quad  \sum_{t'\in T_w(\xi)} \card{\{b\in D\ :\ \tp{w}{b}=t'\}}\ \leq\ c,\\
		& \iff \quad  \sum_{t'\in T_w(\xi)} \mu_w(t')\ \leq\ c.
	\end{align*}}

The final equivalence follows from the fact that each summand strictly less than $(C+1)$, since $c\leq C$. Hence, it follows from the definition that \mbox{$\mu_w(t) = \card{\{b\in D\ :\ \tp{w}{b}=t'\}}$}, for all $t'\in T_w(\xi)$. 
For each index $a\in D$, we define a run $f_a:W\to \Pow{\sub{\phi}}$ by taking
\begin{equation*}
	f_a(w) = \tp{w}{a},
\end{equation*}
for all $w\in W$. We then take $\R=\{r_a : a\in D\}$ to be the set of all such runs through $(W,R,\q)$, with indices from $D$.
Note that there may be many indices in $D$ that correspond to the same run. 
It is straightforward to check that $(W,R,\q,D,\R)$ is a quasimodel for $\phi$.

\bigskip

Conversely, suppose that $\Q=(W,\prec,\q,I,\R)$ is a quasimodel for $\phi$. We define a first-order Kripke model $\M=(\F,D,\dom,\I)$, by taking $\F:=(W,\prec)$, $D=\dom(w):=I$ for all $w\in W$, and 
\begin{equation*}
	\I(w,P_j) = \big\{i\in I : P_j(x) \in r_i(w)\big\},
\end{equation*}
for all predicate symbols $P_j\in \pred$ and $w\in W$. 
It remains to check that $\M$ is a model for $\phi$. 
We claim that 
\begin{equation}\label{IH:quasimodel}\tag{I.H.}
	\M,w \models^i \psi \qquad \iff \qquad \psi \in r_i(w),
\end{equation}
for all $w\in W$, $i\in I$, and $\psi\in \sub{\phi}$.

The case where $\psi$ is an atomic formula follows immediately from the definitions. So suppose that \eqref{IH:quasimodel} holds for all formulas of size $<k$, and let $\psi$ be a formula of size $k$. The cases where $\psi$ is a Boolean combination of smaller formulas follow from the definition of a type, leaving us with two cases:
\begin{itemize}\small
\begin{draftproof}
	\item[--] {\it Case $\psi = P_j(x)$}: This follows immediately from the definition of $P_j^{I(w)}$, since
	{\small \begin{equation*}
	\M,w\models^{i} P_j(x)\; \iff\; i \in \I(w,P_j) \; \iff \; P_j(x) \in r_i(w).
	\end{equation*}}
	\item[--] {\it Case $\psi = \neg \xi$}: We have that
	{\small \begin{equation*}
	\M,w\models^i \neg \xi\; \iff\; \M,w\not\models^i \xi\; \iff[(I.H.)] \; \xi \not\in r(w)\; \iff[{\bf (tp1)}]\; \neg \xi\in r_i(w).
	\end{equation*}}
	\item[--] {\it Case $\psi = (\xi_1 \land \xi_2)$}: We have that
	{\small \begin{align*}
	\M,w\models^i (\xi_1 \land \xi_2) \quad & \iff \quad \M,w\not\models^i \xi_1\ \mbox{and}\ \M,w\not\models^i \xi_2, \\
	& \iff[(I.H.)] \quad \xi_1 \in r_i(w)\ \mbox{and}\ \xi_2 \in r_i(w),\\[3pt]
	& \iff[{\bf (tp2)}] \quad (\xi_1\land \xi_2) \in r_i(w).
	\end{align*}}
\end{draftproof}

	\item[--] {\it Case $\psi = \D\xi$}:  We have that
	{\small \begin{align*}
	\M,w \models^i \D\xi \quad & \iff \quad w \prec v\ \mbox{and}\ \M,v \models^i  \xi, \quad \mbox{for some $v\in W$},\\
	& \iff[(I.H.)] \quad w \prec v\ \mbox{and}\  \xi \in r_i(v),\quad \mbox{for some $v\in W$},\\[3pt]
	& \iff \quad  \D\alpha \in r_i(w) \quad \mbox{by \ref{qm:coherence} and \ref{qm:saturation}.}
	\end{align*}}
	\item[--] {\it Case $\psi = \exists[\leq c]x\ \xi$}: We have that
	{\small
\begin{align*}
	\M,w \models^i (\exists[\leq c] x\ \xi) \quad & 
	\iff[(def)] \quad  \card{\{j\in I\ :\ \M,w\models^j \xi\}}\ \leq\ c\\
	& \iff[(I.H.)] \quad  \card{\{j\in I\ :\ \xi \in r_{j}(w)\}}\ \leq\ c\\
	& \iff \quad  \sum_{t\in T_w(\xi)} \card{\{j\in I\ :\ r_{j}(w) = t\}}\ \leq\ c\\
	& \iff[\ref{qm:runs}] \quad  \sum_{t\in T_w(\xi)} \mu_w(t)\ \leq\ c\\
	& \iff[\ref{qs:saturation}] \quad  (\exists[\leq c] x\ \xi) \in r(w).
\end{align*}	}
The penultimate equivalence follows, again, from the fact that each summand strictly less than $(C+1)$, since $c\leq C$. Hence, it follows from \ref{qm:runs} that \mbox{$\mu_w(t) =  \card{\{j\in I\ :\ r_{j}(w) = t\}}$}, for all $t\in T_w(\xi)$.
\end{itemize}

\smallskip
\noindent
Hence, it follow that $\M,w\models^i \psi$ if and only if $\psi \in r_i(w)$, for all $\psi\in \sub{\phi}$, as required. 
By \ref{qm:witness}, there is some $w_0\in W$ and $t_0\in T_{w_0}$ such that $\phi\in t_0$, while by \ref{qm:runs} we have that there is some $i_0\in I$ such that $r_{i_0}(w_0)=t_0$. Hence, it follows from \eqref{IH:quasimodel} that $\M,w_0\models^{i_0} \phi$, which is to say that $\phi$ is satisfiable with respect to $\QK$, as required.
\end{proof}

\medskip

Hence, to show that the one-variable fragment $\QK\cap \QMLone$ has the \mbox{exponential} \fmp, it is enough to show that every quasimodel for $\phi$ can be transformed into a finite quasimodel in which both $W$ and $I$ are at most exponential in the size~of~$\phi$.

\medskip

\begin{lemma}\label{lem:pruning}
	Let $\phi\in \QMLone$ be an arbitrary formula in one-variable. If $\phi$ has a quasimodel, then $\phi$ has a quasimodel $\Q=(W,\prec,\q,I,\R)$ such that:
	\begin{equation}
		\card{W}\ \leq\ m^2 2^{8(nm)}C
		\qquad \mbox{and} \qquad 
		\card{I}\ \leq\ m2^{5(nm)}C.
	\end{equation}
\end{lemma}
\begin{proof}
Let $\phi\in \QMLone$ and suppose that $\Q=(W,\prec,\q,I,\R)$ is a quasimodel for $\phi$. 
The proof follows two stages: the first involves pruning both the basic structure and the set of runs so that they are both at most exponential in the size of $\phi$. During this stage we inadvertently destroy some of the defining properties of our quasimodel; in particular the saturation condition {\bf (qm4)}. In the second stage we remedy this deficiency by adding multiple `copies' of each quasistate and performing `surgery' on a finite set of runs to repair saturation.

\paragraph{\bf Step 1)}\;
First, it follows from \ref{qm:witness} that there is some $w_0\in W$ and $t_0\in T_{w_0}$ such that $\phi\in t_0$. 
By \ref{qm:runs}, for each $w\in W$ and each $t\in T_w$ we may fix some run $s_{(w,t)}\in \R$ such that $s_{(w,t)}(w) = t$. Take $\subR(w) = \{s_{(w,t)} : t\in T_w\}$ be to the set comprising all such runs, for each $w\in W$. In particular, we note that $\card{\subR(w)} = \card{T_w} \leq 2^n$. 
Furthermore, by \ref{qm:saturation}, for each $\D\alpha\in t$ we may fix some $v=v_{(w,t,\alpha)}\in W$ such that $w \prec v$ and $\alpha \in s_{(w,t)}(v)$. 
We now define inductively a sequence of (finite) subsets $W_k\subseteq W$, for $k=0,\dots, m$, by taking \mbox{$W_0 = \{w_0\}$}, and
\begin{equation*}
	W_{k+1} = \big\{v_{(w,t,\alpha)}\in W\ :\ w\in W_k,\ t \in T_w, \mbox{ and } \D\alpha \in t\big\},
\end{equation*}
for $k<m$. 
We then define a new basic structure $(W',\prec',\q')$, by taking
\begin{equation*}
	W' = \bigcup_{k=0}^m W_k \qquad u \prec'v \iff u \prec v, \qquad \mbox{and} \qquad \q'(u) = \q(u),
\end{equation*}
for all $u,v\in W'$.

Let $\subR=\bigcup \{\subR(w) : w\in W'\}$, and note that $\subR$ is finite since it is a finite union of finite sets of runs. However, $\subR$ need not be plentiful enough to accommodate condition~\ref{qm:runs}. Hence we must extend $\subR$ to a `small' subset $\R'$ of $\R$ by choosing sufficiently many runs so as to satisfy \ref{qm:runs}. 

More precisely, for each $w\in W'$, $t\in T_w$ and $m<\mu_w(t)$ we can fix some $r_{(w,t,m)}\in \R$ such that $r_{(w,t,m)}(w) = t$, and $r_{(w,t,m)}\not=r_{(w,t,m')}$ for $m\not=m'$. The existence of sufficiently many such runs is guaranteed by \ref{qm:runs}. Furthermore, we may assume without any loss of generality that $r_{(w,t,0)}=s_{(w,t)}\in \subR(w)$, as defined above. 
We may then take 
\begin{eqnarray*}
	I' & = & \big\{(w,t,m)\ :\ w\in W',\ t\in T_w \mbox{ and } m< \mu_w(t)\big\},\\[5pt]
	\R' & = & \big\{r_{(w,t,m)} \in \R : (w,t,m) \in I'\big\},
\end{eqnarray*}
and define $\Q'=(W',\prec',\q',\R')$. 
We note that:
\begin{eqnarray*}
\label{eq:size1a}
\card{W'} & \leq & \card{W_0} + \dots + \card{W_m} \leq m2^{3(nm)},\\[5pt]
\label{eq:size1b}
\card{I'} & \leq & \card{W'} \cdot \max_{w\in W} \card{T_w} \cdot (C+1) \leq m2^{5(nm)}C.
\end{eqnarray*}
Furthermore, from our construction we have that $\Q'$ satisfies each of the conditions \ref{qm:basic}, \ref{qm:witness}, \ref{qm:coherence}, and \ref{qm:runs}, as can be easily verified. However $\Q'$ fails to satisfy the saturation condition \ref{qm:saturation}. To remedy this, we diverge from the techniques of~\cite{Wolter2000,GKWZ03} by extending our basic structure with not one but \emph{multiple} `copies' of each quasistate; each associated with a given transposition of runs.

\paragraph{\bf Step 2)}\;
Let $\Sym{I'}$ denote the set of all permutations $\sigma:I' \to I'$ on the set of indices $I'$, with $\id\in \Sym{I'}$ denoting the identity function. 
For each $w\in W'$ and each $i\in I'$, let $\tau_{(w,i)}\in \Sym{I'}$ denote the permutation that transposes $r_i$ and $s_{(w,t)}\in \subR(w)$, where $t=r_i(w)$. Let $\Trans{\R'}{w} = \{\tau_{(w,i)} : i\in I'\}$ denote the set of all such transpositions. In particular, we have that $\card{\Trans{I'}{w}}\leq \card{I'}$ is at most exponential in the size of $\phi$.

\smallskip

%

In what follows, we construct a new basic structure based on some `small' subset of $W'\times \Sym{I'}$. Naturally, we cannot construct a basic structure out of the set of all pairs from $W'\times \Sym{I'}$ if we are to insist on an exponential upper bound on the size of the quasimodel, since $\card{\Sym{I'}} = \card{I'}!$. 
Instead, for each $(w,\sigma)\in W'\times \Sym{I'}$, we may define a small set of \emph{successors} $S(w,\sigma)\subseteq W'\times \Sym{I'}$, by taking:
\begin{equation*}
	S(w,\sigma) = \{(v,\sigma')\ :\ w \prec v\ \mbox{and}\ \sigma'=(\tau\circ \sigma)\ \mbox{for some}\ \tau \in \Trans{I'}{w}\},
\end{equation*}
for all $w\in W'$ and $\sigma\in \Sym{I'}$. 
%
%
We construct a new sequence of sets \mbox{$W_k'\subseteq W'\times \Sym{I'}$}, for $k=0, \dots, m$, by taking
\begin{equation*}
	W_0' = \{(w_0,\id)\} \qquad \mbox{and} \qquad W_{k+1}' = \bigcup \big\{S(w,\sigma)\ :\ (w,\sigma)\in W_k'\big\},
\end{equation*}
for $k<m$. Define a new basic structure $(W'',\prec'',\q'')$, by taking:
\begin{equation*}
	W'' = \bigcup_{k=0}^m W_k', \qquad\qquad (u,\sigma)\prec''(v,\rho) \iff (v,\rho)\in S(u,\sigma)
\end{equation*}
and $\q''(u,\sigma) = \q'(u)$, for all $(u,\sigma),(v,\rho)\in W''$.

Finally, for each run $i\in I'$ we define a new run $\run{r}_i$ through $(W'',\prec'',\q'')$, by taking $\run{r}_i(w,\sigma) = r_{\sigma(i)}(w)$, for all $(w,\sigma)\in W''$. 
That is to say that the new run $\run{r}_i$ behaves at $(w,\sigma)\in W''$ as $r_{\sigma(i)}\in \R'$ does at $w\in W'$. 
Take $I''=I'$ and let $\R''=\{\run{r}_i : i \in I''\}$ be the set of all such runs. We may then define a new quasimodel $\Q''=(W'',\prec'',\q'',I'',\R'')$, and note that
\begin{equation*}
\card{W''} \leq \card{W'}\cdot \card{I'} \leq m^2 2^{8(nm)}C \qquad \mbox{and} \qquad \card{I''}=\card{I'} \leq m2^{5(nm)}C,
\end{equation*}
both remaining at most exponential in the size of $\phi$, as required. All that remains is to show that $\Q''$ is, indeed, a quasimodel for $\phi$.

\begin{itemize}\small
	\item[---] It follows from the construction that $\phi \in t_0$ for some $t_0\in T_{w_0}=T_{(w_0,\id)}$, where $(w_0,\id)\in W'$, as required for\ref{qm:witness}.
	\item[---] For \ref{qm:coherence}, suppose that $i\in I''$, $(w,\sigma),(v,\rho)\in W''$ and $\D\alpha\in \sub{\phi}$ are such that $(w,\sigma)\prec'' (v,\rho)$ and $\alpha\in \run{r}(v,\rho)$.
	
	By definition we have that $(v,\rho)\in S(w,\sigma)$, which is to say that $w \prec v$ and $\rho=\tau\circ \sigma$ for some transposition $\tau\in \Trans{I'}{w}$. Hence we have that 
\begin{equation*}
	\alpha \in \run{r}_i(v,\rho)\ =\ r_{\rho(i)}(v)\ =\ r_{(\tau\circ \sigma)(i)}(v)\ =\ r_{j}(v),
\end{equation*}
where $j=\tau(\sigma(i))\in I'$. Since $\Q'$ is coherent and $w \prec' v$, we have that \mbox{$\D\alpha\in r_{j}(w)$}. 
However, we have that $\tau\in \Trans{I'}{w}$ and hence by definition $r_{\tau(\sigma(i))}(w)=r_{\sigma(i)}(w)$, since $\tau$ transposes only runs that coincide at $w$.  In particular, we have that $\D\alpha\in r_{\sigma(i)}(w)$, which is to say that $\D\alpha\in \run{r}_i(w,\sigma)$, as required.

	\item[---] For \ref{qm:saturation}, suppose that $i\in I'$, $(w,\sigma)\in W''$ and \mbox{$\D\alpha\in \sub{\phi}$} are such that $\D\alpha\in \run{r}_i(w,\sigma)$. This is to say that $\D\alpha \in r_{\sigma(i)}(w)$. Let $t=r_{\sigma(i)}(w)$ and let $s_{(w,t)}\in \subR(w)$ be such that $s_{(w,t)}(w)=t$. By construction there is some $v=v_{(w,t,\alpha)}\in W'$ such that $w \prec' v$ and $\alpha \in s_{(w,t)}(v)$. 
	
Let $\tau=\tau_{(w,\sigma(i))}\in \Trans{I'}{w}$ be the transposition that swaps $r_{\sigma(i)}\in \R'$ and $s_{(w,t)}\in \subR(w)$. 
It follows from the construction that there is some $(v,\tau\circ\sigma)\in S(w,\sigma)\subseteq W''$ such that
\begin{equation*}
	\alpha\ \in \ s_{(w,t)}(v)\ =\ r_{\tau(\sigma(i))}(v)\ =\ r_{(\tau\circ \sigma)(i)}(v)\ =\ \run{r}_i(v,\tau\circ \sigma),
\end{equation*}
and $(w,\sigma)\prec'' (v,\tau\circ \sigma)$, as required.
	\item[---] For \ref{qm:runs}, suppose that $(w,\sigma)\in W''$ and $t\in T_{(w,\sigma)} = T_w$ and consider the following sets:
	\begin{equation*}
		X = \{i\in \I' : \run{r}_i(w,\sigma) = t\}\qquad \mbox{and} \qquad  Y= \{i\in I' : r_i(w) = t\}.
	\end{equation*}
	
	Note that the $\sigma$ maps bijectively from $X$ onto $Y$, since by definition $\run{r}_i(w,\sigma)=r_{\sigma(i}(w)$. Hence $i\in X$ if and only if $\sigma(i)\in Y$, and thus $\card{X}=\card{Y}$.
	
	%
That is to say that the number of runs passing through each type remains unaffected by Step 2 of our construction. It then follows from the definitions that
	\begin{equation*}
		\mu_{(w,\sigma)}(t)\ =\ \mu_w(t)\ =\ \min(\card{Y},C+1)\ =\ \min(\card{X},C+1)
	\end{equation*}
	as required.	
\end{itemize}	

\smallskip

\noindent
Thus completes the proof of Lemma~\ref{lem:pruning}.
\end{proof}

\medskip

Theorem~\ref{thm:QKfmp} now follows from Lemmas~\ref{lem:qm}--\ref{lem:pruning}, and hence the one-variable fragment $\QK\cap \QMLone$ has the exponential \fmp. Consequently, as described above, we may exploit this exponential \fmp\ to answer the satisfiability problem in \NExpTime, thereby completing the proof of Theorem~\ref{thm:QK-unbounded}.
%



It follows that the satisfiability problems for each of the one-variable fragments $\QK\cap \QMLone[\ell]$ with finitely bounded quantifier subscripts are similarly \NExpTime-complete; sandwiched, as they are, between the unbounded fragment $\QK\cap\QMLone$ and the \NExpTime-hard one-variable counting-free fragment $\QK\cap\QMLcf^1$~\cite{Marx1999}. 
Furthermore, the satisfiability problem for the fragment $\QK\cap \QMLoneun$, in which the counting quantifiers are encoded as unary strings, can be polynomially reduced to that of $\QK\cap \QMLone$ by transcribing the subscripts into binary. Consequently, it shares the same \NExpTime\ upper-bound.

\begin{corollary}
\label{cor:QK-bounded}
\label{cor:QK-unary}
	The satisfiability problem for each of the fragments $\QK\cap \QMLone[\ell]$ and $\QK\cap \QMLoneun$ is \NExpTime-complete, for $\ell<\omega$.
\end{corollary}

\section{Logics with unbounded quantifiers over expanding or decreasing domains}
\label{sec:unbounded}

\newcommand{\E}{\mathcal{E}}

It is well-established that the satisfiability problems for the counting-free fragments of $\QKexp$ and $\QKdec$ are both polynomially reducible to that of $\QK$, by relativizing the domain function with an auxiliary monadic predicate symbol $\E\in \pred$, demarcating those domain objects that \emph{actually exist}~\cite{FittingMendelsohn1998,Wolter1998}. 

The same is true over the language with counting quantifiers, where the same trick can also be employed. Let $\phi\in \QML$ be an arbitrary formula first-order modal formula with counting quantifiers, and let $\E\in \pred$ be a fresh predicate symbol not occurring in $\phi$. We define the \emph{relativization} $\phi^\E$ via the function $(\cdot)^\E :\sub{\phi} \to \QML$, defined inductively, by taking
\begin{gather*}
	P(x_1,\dots, x_n)^\E := P(x_1,\dots,x_n), \quad 
	(\psi_1 \land \psi_2)^\E := (\psi_1^\E \land \psi_2^\E)\\[5pt]
		(\neg\psi)^\E := \neg \psi^\E, \quad (\D\psi)^\E := \D\psi^\E, \quad
	\mbox{and} \quad
	(\exists[\leq c]x\ \psi)^\E := \exists[\leq c] x \big(\E(x) \land \psi^\E\big).
\end{gather*}
We may, thereby, reduce the satisfiability problem for $\QKexp\cap \QML$ and $\QKdec\cap \QML$ to that of $\QK\cap\QML$, by specifying the expanding or decreasing nature of $\E$, by defining:
\newcommand{\zetaexp}{\zeta_\ex}
\newcommand{\zetadec}{\zeta_\dc}
\begin{eqnarray*}
	\zetaexp  &:= & \B^{\leq m}\forall x (\E(x) \to \B \E(x)), \\
	\zetadec & := &  \forall x\B^{\leq m} (\D\E(x) \to \E(x)),
\end{eqnarray*}
respectively, where $m=\md{\phi}$ is the modal depth of $\phi$. 
\begin{proposition}
Let $\phi\in \QML$ be an arbitrary first-order modal formula. Then we have the following equivalences:
\begin{itemize}
	\item[\rm(i)] $\phi$ is satisfiable with respect to $\QKexp$ if and only if $(\zetaexp \land \phi^\E)$ is satisfiable with respect to $\QK$,
	\item[\rm(ii)] $\phi$ is satisfiable with respect to $\QKdec$ if and only if $(\zetadec \land \phi^\E)$ is satisfiable with respect to $\QK$,
\end{itemize}
\end{proposition}
\begin{proof}
The proof is via a routine induction, analogous to that of \cite[Proposition 2.4]{Wolter1998}.
\end{proof}

This yields the following immediate consequence of Theorem~\ref{thm:QK-unbounded}.


\begin{corollary}
\label{cor:unbounded}
	The satisfiability problem for the one-variable fragment $L\cap \QMLone$ with unbounded quantifiers is decidable in \NExpTime, for $L\in \{\QKexp,\QKdec\}$.
\end{corollary}

What prevents us from encoding constant domains within varying domain models is our inability to prevent the domains from expanding or contracting beyond the scope of our formulas. However, the addition of counting quantifiers allows us to make specific demands on the size of each domain.

The following theorem provides a matching lower bound, by exploiting the exponential domain property of $\QK\cap \QMLone$, proved in Theorem~\ref{thm:QK-unbounded}, together with the succinct binary encoding of counting quantifiers. The intuition is that with unbounded quantifiers, we can specify the exact size of the first-order domains, forcing them to be remain constant in all possible worlds lie within the scope of the formula. The binary encoding allows us to succinctly specify the (possibly) exponential size of the domains.


\begin{theorem}
\label{thm:QKexp-unbounded}
\label{thm:QKdec-unbounded}
The satisfiability problem for the one-variable fragment $L\cap \QMLone$ with unbounded counting quantifiers is \NExpTime-complete, for $L\in \{\QKexp,\QKdec\}$.
\end{theorem}
\begin{proof}
The upper-bound has already been established by Corollary~\ref{cor:unbounded}. 
The proof for the lower-bound is via a reduction from the satisfiability problem for the one-variable fragment $\QK\cap \QMLone$, over constant domains. 

To this end, let $\phi\in \QMLone$ be an arbitrary formula in one variable, whose counting quantifiers are encoded as binary strings. 
Since the one-variable fragment $\QK\cap\QMLone$ has the exponential finite domain property, there is some monotonic function $f(n)\in O\big(2^{n^k}\big)$ such that if $\phi$ is satisfiable with respect to $\QK$ then $\phi$ is satisfiable in a model in which the size of each first-order domain does not exceed $N=f(\size{\phi})$.

\smallskip

\noindent
\newcommand{\BM}[1][m]{\forall x \B^{\leq #1} \forall x\ }
We define $\zeta$ to be the conjunction of $\zetaexp$, $\zetadec$ and the following formula:
\begin{align}
\label{eq:zeta1}
 & \BM\big(\exists[\leq N] x \top(x) \land \neg \exists[\leq (N-1)] x \top(x)\big)
\end{align}
where $\top(x) := P_0(x) \lor \neg P_0(x)$. Note that the size of $\zeta$ is at most logarithmic in the size of $N$, owing to the binary encoding of the quantifier subscripts.

%
We claim that $\phi$ is satisfiable with respect to $\QK$ if and only if $(\zeta \land \phi^\E)$ is satisfiable with respect to $\bfL$, for $\bfL\in \{\QKexp,\QKdec\}$.

\begin{itemize}
	\lefttoright Suppose $\M,r\models^{a_0} (\zeta\land \phi^\E)$, for some first-order Kripke model $\M=(\F,D,\dom,\I)$, with $a_0\in \dom(r)$, where $\dom$ is either \emph{expanding} or \emph{decreasing}. Without loss of generality, we may assume that $\F=(W,R)$ is an irreflexive, intransitive tree of depth $\leq m$. Whence, by \eqref{eq:zeta1}, it follows that $\M,w\models^{a} \exists[\leq N] \top(x) \land \neg\exists[\leq (N-1)] \top(x)$, for all $w\in W$ and $a\in \dom(w)$. 
	 Consequently, we have that $\card{\dom(w)} = N$, for all $w\in W$, and since $\dom$ is assumed to be either increasing or decreasing, we must have that $\dom(u)=\dom(v)$, for all $u,v\in W$. It then follows from $\zetaexp$ and $\zetadec$ that $\I(u,\E) = \I(v,\E)$, for all $w,v\in W$.
	
	We define a new \emph{constant domain} model $\M'=(\F,D,\dom',\I')$ over $\F$, by taking \mbox{$\dom'(w) = \I(w,\E)$}, for all $w\in W$ and $\I'(w,P_i) = \I(w,P_i)$, for all predicate symbols $P_i$ occurring in $\phi$.
	
	We prove, by induction on the length of $\psi\in\sub{\phi}$ that
	\begin{equation}
	\tag{I.H.1}
	\label{IH:unboundedLR}
		\M',w\models^a \psi \qquad \iff \qquad \M,w\models^a \psi^\E,
	\end{equation}
for all $w\in W$ and $a\in \dom'(w)$.

The cases where $\psi$ is an atomic formula or a Boolean combination of smaller formulas are straightforward and follow from the definitions. So suppose that $\psi$ is of the form $\D\xi$ or $\exists[\leq c]x \xi$, for some $\xi\in \sub{\phi}$ and $c\in \bbN$. In which case we have the following:
	
	\begin{itemize}\small
	\begin{draftproof}
	\item[--] \emph{Case $\psi = P_i(x)$}:\quad We have that
	\begin{equation*}
		\M',w\models^a P_i(x) 
			\;\; \iff \;\; a\in \I'(w,P_i) 
			\;\; \iff[(def)] \;\; a\in \I(w,P_i) 
			\;\; \iff \;\; \M,w\models^a P_i(x).
	\end{equation*}	
	
	\item[--] \emph{Case $\psi = \neg\xi$}:\quad We have that
	\begin{eqnarray*}
		\M',w\models^a \neg \xi 
			& \iff & \M',w\not\models^a\xi \\
			& \iff[\eqref{}] & \M,w\not\models^a\xi^\E \\
			& \iff & \M,w\models^a\neg\xi^\E\\
			& \iff[(def)] & \M,w\models^a(\neg\xi)^\E
	\end{eqnarray*}	
	
	\item[--] \emph{Case $\psi = (\xi_1\land \xi_2)$}:\quad We have that
	\begin{eqnarray*}
		\M',w\models^a (\xi_1\land \xi_2)
			& \iff & \M',w\not\models^a\xi_1 \mbox{ and } \M',w\not\models^a\xi_2 \\
			& \iff[\eqref{}] & \M,w\not\models^a\xi_1^\E \mbox{ and } \M,w\not\models^a\xi_2^\E \\
			& \iff & \M,w\models^a(\xi_1^\E\land \xi_2^\E)\\
			& \iff[(def)] & \M,w\models^a(\xi_1\land \xi_2)^\E
	\end{eqnarray*}	
	\end{draftproof}	
	
	\item[--] \emph{Case $\psi = \D\xi$}:\quad We have that
	\begin{eqnarray*}
		\M',w\models^a \D \xi 
			& \iff & \exists v\in W;\ w R v \mbox{ and } \M',v\models^a\xi \\
			& \iff[\eqref{IH:unboundedLR}] & \exists v\in W;\ w R v \mbox{ and } \M,v\models^a\xi^\E \\
			& \iff & \M,w\models^a \D\xi^\E\\
			& \iff[(def)] & \M,w\models^a (\D\xi)^\E
	\end{eqnarray*}

	\item[--] \emph{Case $\psi = \exists[\leq c]x \ \xi$}:\quad We have that
	\begin{eqnarray*}
		\M',w\models^a \exists[\leq c] x\ \xi
			& \iff & \card{\big\{b\in \dom'(w) : \M',w\models^b \xi\big\}}\leq c\\
			& \iff[\eqref{IH:unboundedLR}] & \card{\big\{b\in \dom'(w): \M,w\models^b \xi^\E\big\}}\leq c\\
			& \iff[(def)] & \card{\big\{b\in \dom(w) : b\in E^{I(w)}\ \mbox{and}\ \M,w\models^b \xi^\E\big\}}\leq c\\
			& \iff & \card{\big\{b\in \dom(w) : \M,w\models^b \E(x) \land \xi^\E\big\}}\leq c\\
			& \iff & \M,w\models^a \exists[\leq c]x\ \big(\E(x) \land \xi^\E\big)\\
			& \iff[(def)] & \M,w\models^a \left(\exists[\leq c] x \ \xi\right)^\E
	\end{eqnarray*}
	
	\end{itemize}
	
Hence, we have that $\M',w\models^a \psi$ if and only if $\M,w\models^a \psi^\E$, for all $\psi\in \sub{\phi}$. In particular, we have that $\M',r\models^{a_0} \phi$, which is to say that $\phi$ is satisfiable with respect to $\QK$.
	
	\righttoleft Suppose that $\M,r\models^{a_0}\phi$ for some first-order \emph{constant domain} model $\M=(\F,D,\dom,\I)$, with $a_0\in \dom(r)$. Again, without any loss of generality, we may assume that $\F=(W,R)$ is an irreflexive intransitive tree of depth $\leq m$, with root $r\in W$.
	Moreover, since $\QK\cap \QMLone$ has the $f(n)$-sized domain property, we may assume that $\card{D}\leq N$, where $N=f(\size{\phi})$ is an most exponential in the size of $\phi$. 
	Hence, there must be some injective enumeration $\iota:D\to \{0,1,\dots, N-1\}$ of $D$, assigning to each object $a\in D$, a unique index $\iota(a)< N$.
	\renewcommand{\Im}[1]{\mathrm{Image}(#1)}
	We define a new model $\M'=(\F,D',\dom',\I')$ by taking $D' = \{0,1,2,\dots,N-1\}$, and $\dom'(w)=D'$, for all $w\in W$ and, 
	\begin{eqnarray*}
	\I'(w,P_i) &=& \{\iota(a) < N : a\in \I(w,P_i)\},\\
	\I'(w,\E) &=& \{\iota(a) < N : a\in D\}, 
\end{eqnarray*}
for all predicate symbols $P_i\in \pred$ occurring in $\phi$.\\

We prove, by induction on the length of $\psi\in \sub{\phi}$ that
	\begin{equation}
	\tag{I.H.2}
	\label{IH:unboundedRL}
		\M',w\models^{\iota(a)} \psi^\E \qquad \iff \qquad \M,w\models^a \psi,
	\end{equation}
for all $w\in W$ and $a\in \dom(w)$.
	
	
	\begin{itemize}\small
	\begin{draftproof}
	\item[--] \emph{Case $\psi = P_i(x)$}:\quad We have that
	\begin{equation*}
		\M',w\models^a P_i(x) 
			\;\; \iff \;\; a\in \I'(w,P_i) 
			\;\; \iff[(def)] \;\; a\in \I(w,P_i)
			\;\; \iff \;\; \M,w\models^a P_i(x).
	\end{equation*}	
	
	\item[--] \emph{Case $\psi = \neg\xi$}:\quad We have that
	\begin{eqnarray*}
		\M',w\models^a (\neg \xi)^\E
			& \iff[(def)] & \M',w\models^a\neg\xi^\E \\
			& \iff & \M',w\not\models^a\xi^\E \\
			& \iff[\eqref{IH:unboundedRL}] & \M,w\not\models^a\xi \\
			& \iff & \M,w\models^a\neg\xi
	\end{eqnarray*}	
	
	\item[--] \emph{Case $\psi = (\xi_1\land \xi_2)$}:\quad We have that
	\begin{eqnarray*}
		\M',w\models^a (\xi_1\land \xi_2)^\E
			& \iff[(def)] & \M',w\models^a (\xi_1^\E\land \xi_2^\E)\\
			& \iff & \M',w\not\models^a\xi_1^\E \mbox{ and } \M',w\not\models^a\xi_2^\E \\
			& \iff[\eqref{IH:unboundedRL}] & \M,w\not\models^a\xi_1 \mbox{ and } \M,w\not\models^a\xi_2 \\
			& \iff & \M,w\models^a(\xi_1\land \xi_2)
	\end{eqnarray*}	
	\end{draftproof}
	
	\item[--] \emph{Case $\psi = \D\xi$}:\quad We have that
	\begin{eqnarray*}
		\M',w\models^a (\D \xi)^\E
			& \iff[(def)] & \M',w\models^a \D\xi^\E\\ 
			& \iff & \exists v\in W;\ w R v \mbox{ and } \M',v\models^a\xi^\E \\
			& \iff[\eqref{IH:unboundedRL}] & \exists v\in W;\ w R v \mbox{ and } \M,v\models^a\xi \\
			& \iff & \M,w\models^a\D\xi
	\end{eqnarray*}

	\item \emph{Case $\psi=(\exists[\leq c]x \ \xi)$}:\quad We have that 
	\begin{eqnarray*}
		\M',w\models^a (\exists[\leq c] x\ \xi)^\E
			& \iff & \M',w\models^a \exists[\leq c] x \big(\E(x) \land \xi^\E\big)\\
			& \iff & \card{\big\{b\in \dom'(w) : \M',w\models^b \E(x) \land \xi^\E\big\}}\leq c\\
			& \iff & \card{\big\{b\in \dom'(w): \M',w\models^b \E(x)\ \mbox{and}\ \M',w\models^{b} \xi^\E\big\}}\leq c\\
			& \iff[(def)] & \card{\big\{b\in \dom'(w) : \M',w\models^{\iota(b)} \xi^\E\big\}}\leq c\\
			& \iff[\eqref{IH:unboundedRL}] & \card{\big\{\iota(b)\in \dom(w) : \M,w\models^{\iota(b)} \xi\big\}}\leq c\\
	& \iff & \M,w\models^a \exists[\leq c]x\ \xi.\\
	\end{eqnarray*}
	\end{itemize}
Hence, we have that $\M',\models^a \psi^\E$ if and only if $\M,w\models^a \psi$, for all $\psi\in \sub{\phi}$. In particular, we have that $\M',r\models^{a_0} \phi$. 
Furthermore, it is straightforward to check that $\M',r\models^{a_0} \zeta$, since each $\dom'(w)$ contains precisely $N$ elements, and $\I'(u,\E)=\I'(v,\E)$, for all $u,v\in W$. 
Hence, we have that $\M,r\models^{a_0}(\zeta \land \phi^\E)$, which is to say that $(\zeta \land \phi^\E)$ is satisfiable with respect to $\QK$, and consequently with respect to both $\QKexp$ and $\QKdec$.
\end{itemize}

Since the satisfiability problem for the one-variable fragment $\QK\cap \QMLone$ is \NExpTime-hard, so too must be the satisfiability problem for the one-variable fragment \mbox{$\bfL\cap\QMLone$}, for both $\bfL\in \{\QKexp,\QKdec\}$, as required.
\end{proof}

Note, above, that $\zeta$ is not restricted to any quantifier bounded fragment $\QMLone[\ell]$, for $\ell<\omega$. 
Furthermore, this reduction relies heavily on the succinct binary encoding of the subscripts, without which we would be unable to specify the exponential size of the domains while maintaining the polynomial bound on the size of $\zeta$.

\smallskip

In the proceeding sections, we shall see that both of these assumptions are crucial to the success of Theorem~\ref{thm:QKexp-unbounded}, and without which the computational complexity of the satisfiability problem can be greatly reduced.

\section{Logics with bounded quantifiers over expanding domains}
\label{sec:expanding}

In this section we show that placing \emph{any} finite bound on the quantifier subscripts significantly reduces the computational complexity of the satisfiability problem for the resulting fragment, from \NExpTime-complete to \PSpace-complete.

Similarly, if we were to encode the quantifier subscripts as unary strings---a practice that is common within other branches of logic, such as description logics~\cite{Baader2003,Hustadt2004}---then we also observe the same reduction in complexity. 
In both cases, the results hinges on the fact that the domains for any first-order Kripke model for $\phi$ can be chosen to be at most polynomial in the size of $\phi$ and $C=\cpt{\phi}$. Note that $C$ is bounded by some fixed constant for $\phi\in \QMLone[\ell]$ and is at most linear in $\phi$ when encoded as a unary string.

\begin{lemma}
\label{lem:polysize}
	Given an arbitrary first-order modal formula $\phi\in \QMLone$. If $\phi$ is satisfiable with respect to $\QKexp$ then, $\phi$ can be satisfied in a model $\M=(\F,D,\dom,I)$ such that,
	\begin{equation}
	\label{eq:polysize}
		\card{\dom(w)}\ \leq\ nm(C+1), \qquad \mbox{for all $w\in W$,}
	\end{equation}
where $n=\card{\sub{\phi}}$, $m=\md{\phi}$ and $C=\cpt{\phi}$.
\end{lemma}
\begin{proof}
Suppose that $\phi\in \QMLone[\ell]$ is satisfiable with respect to $\QK$. Then $\M,r\models^{a_0} \phi$ for some expanding first-order Kripke model $\M=(\F,D,\dom,\I)$, where $\F=(W,R)$ is an irreflexive, intransitive tree of depth at most $m=\md{\phi}$ with root $r\in W$. 
For each $w\in W$, and $\psi \in \sub{\phi}$, let 
\begin{equation*}
	n(w,\psi) \ = \ \min\left(\card{\{a\in \dom(w) : \psi\in \tp{w}{a}\}}, \ C+1\right)
\end{equation*}
count the number of elements in $\dom(w)$ satisfying $\psi$, up to a maximum of $(C+1)$, beyond which $\phi$ lacks that vocabulary to discern. For each $i< n(w,\phi)$, choose some fixed witness $a_{\psi,i}\in \dom(w)$ such that $\psi\in \tp{w}{a_{\psi,i}}$, and $a_{\psi,i}\not=a_{\psi,j}$, for $i\not=j$. 
We then define a `small' subset $D_w\subseteq \dom(w)$, by taking
\begin{equation*}
	D_w \ = \ \{a_{\psi,i}\in \dom(w) \ : \ \psi\in \sub{\phi} \mbox{ and } i<n(w,\psi)\} 
\end{equation*}
for all $w\in W$. It follows that $\card{D_w}\leq n(C+1)$, for all $w\in W$.
%
We may now define a new domain function $\dom':W\to D$, inductively, by taking
\begin{equation*}
	\dom'(r) \ = \ \{a_0\}\cup D_r
	 \qquad \mbox{and} \qquad 
	 \dom'(w) = \dom(v) \cup D_w
\end{equation*}
where $v$ is the (unique) predecessor of $w$, so that  $\dom'(v) \subseteq \dom'(w)$ whenever $v R w$. Furthermore, we note that $\card{\dom'(w)} \leq \card{\dom'(v)} + n(C+1)$, from which we can may infer that $\card{\dom'(w)} \leq nm(C+1)$, as required.

\smallskip
\noindent
We define a new model $\M'=(\F,D,\dom',\I')$, over $\F$ by taking $\I'(w,P_i) = \I(w,P_i) \cap \dom'(w)$, for all $w\in W$ and all predicate symbols $P_i\in\pred$ occurring in $\phi$. 

\smallskip

\noindent
It follows from a straightforward induction on this size of $\psi\in \sub{\phi}$, that
\begin{equation*}
	\M',w\models^a \psi \qquad \iff \qquad \M,w\models^a \psi
\end{equation*}
for all $w\in W$ and $a\in \dom'(w)$. 
The only non-trivial case is where $\psi$ is of the form $\exists[\leq c]x \xi$, for some $\xi\in \sub{\phi}$ and $c\in \bbN$. In which case we have that:
\begin{itemize}\small
\begin{draftproof}
	\item[--] \emph{Case $\psi=P_i(x)$}: \; We have that
	\begin{eqnarray*}
		\M',w\models^a P_i(x) & \iff & a\in \I'(w,P_i)\\
		& \iff & a \in \I(w,P_i) \cap \dom'(w)\\
		& \iff & \M,w\models^a P_i(x)
	\end{eqnarray*}
	
	\item[--] \emph{Case $\psi=\neg\xi$}: \; We have that
	\begin{eqnarray*}
		\M',w\models^a \neg \xi & \iff & \M',w\not\models^a \xi\\
		& \iff[(I.H.)] & \M,w\not\models^a \xi\\
		& \iff & \M,w\models^a \neg \xi
	\end{eqnarray*}
	
	\item[--] \emph{Case $\psi=(\xi_1\land \xi_2)$}: \; We have that
	\begin{eqnarray*}
		\M',w\models^a \neg (\xi_1 \land \xi_2) & \iff & \M',w\models^a \xi_1 \mbox{ and } \M',w\models^a \xi_2\\
		& \iff[(I.H.)] & \M,w\models^a \xi_1 \mbox{ and } \M,w\models^a \xi_2\\
		& \iff & \M,w\models^a (\xi_1 \land \xi_2)
	\end{eqnarray*}

	\item[--] \emph{Case $\psi=\D\xi$}: \; We have that
	\begin{eqnarray*}
		\M',w\models^a \D \xi & \iff & \exists v\in W;\; w R v \mbox{ and }\M',v\models^a \xi\\
		& \iff[(I.H.)] & \exists v\in W;\; w R v \mbox{ and }\M,v\models^a \xi\\
		& \iff & \M,w\models^a \D \xi
	\end{eqnarray*}
	\end{draftproof}
	
	\item[--] \emph{Case $\psi=(\exists[\leq c] x\ \xi)$}: \;
	\begin{eqnarray*}
		\M',w\models^a (\exists[\leq c]x\ \xi)
			& \iff & \card{\{b\in D_w' \ : \ \M',w\models^b \xi\}} \leq c\\
			& \iff[(I.H.)] & \card{\{b\in D_w' \ : \ \M,w\models^b \xi\}} \leq c\\
			& \iff & \card{\{b\in D_w \ : \ \M,w\models^b \xi\}} \leq c \qquad \mbox{since $c<C+1$}\\
			& \iff & \M,w\models^a (\exists[\leq c] x \ \xi)
	\end{eqnarray*}
\end{itemize}
Hence, it follows that $\M',r\models^{a_0}\phi$, as required.
\end{proof}

Note that, in general, the bound given in \eqref{eq:polysize} is exponential in the size of $\phi$ owing to the binary encoding of subscripts. However, if $\phi\in \QMLone[\ell]$ belongs to any of the quantifier bound fragments, then $C<\ell$ is at most constant. Likewise, if the quantifiers are encoded as unary strings, then $C$ is also at most linear in the size of $\phi$. 

\begin{corollary}
	Each of the fragments $\QKexp\cap \QMLone[\ell]$ and $\QKexp\cap\QMLoneun$ possess the poly-size domain property.
\end{corollary}  

We can exploit this polysize domain property to construct a tableau algorithm, in the style of Ladner's $\Kworld$ algorithm~\cite{Ladner1977}, that answers the satisfiability problem using at most polynomial space. 
%
%
For succinctness, the approach taken here follows the presentation given by Spaan~\cite{Spaan1993}.

\newcommand{\indexset}[1]{I_{#1}}
\newcommand{\bbZ}{\mathbb{Z}}
\begin{theorem}
\label{thm:QKexp-bounded}
Let $\mathcal{K}\subseteq \QML$ be any fragment of $\QML$ such that $\QKexp\cap \mathcal{K}$ has the pol-sized domain property. Then the satisfiability problem for $\QK\cap\mathcal{K}$ is \PSpace-complete.
\end{theorem}
\begin{theorem}
\label{thm:QKexp-bounded}
Let $\mathcal{K}\subseteq \QML$ be any fragment of $\QML$ such that $\QKexp\cap \mathcal{K}$ has the pol-sized domain property. Then the satisfiability problem for $\QK\cap\mathcal{K}$ is \PSpace-complete.
\end{theorem}
\begin{proof}
In what follows, let $\indexset{n}=\{0,\dots, (n-1)\}$ denote the set of natural numbers $<n$.
We define a recursive function $\QKworld$ that takes five parameters: a natural number $k\in\bbN$ controlling the recursion depth and ensuring termination, a set of fomulas $\Sigma$, two positive integers $N,t\in \bbZ^+$, and a labelling function $\lambda:\indexset{t}\to \Pow{\Sigma}$ associating each integer $i\leq t$ with a set of formulas $\lambda(i)\subseteq \Sigma$. 
Note that we can encode $\lambda$ succintly as a list of pairs
\begin{equation*}
G(\lambda) = \{(a,\psi) : \psi\in \lambda(a)\} \subseteq \indexset{t}\times\Sigma,
\end{equation*}
whose size is at most $t\times \card{\Sigma}$. We shall refer to $G(\lambda)$ as the `graph' of $\lambda$.

\newcommand{\true}{\textbf{TRUE}}
\newcommand{\false}{\textbf{FALSE}}
\noindent
For $k<\omega$, the function $\QKworld[\Sigma,N](k,t,\lambda)$ returns \true\ if and only if the following conditions are met\footnote{For completeness, a more detailed description of how $\QKworld$ may be implemented is described in Appendix~\ref{app:QKworld}.}:
\begin{enumerate}[label={\bf (tab\arabic*)}]\itemindent=1em\itemsep=7pt
	\item \label{tab:boolean} For each $i<t$, we require that $\lambda(i)$ is Boolean saturated subset of $\Sigma$, which is to say that:
	\begin{itemize}\itemsep=5pt
		\item[(i)] $\neg\psi\in \lambda(i)$ if and only if $\psi\not\in \lambda(i)$, for all $\neg\psi\in \Sigma$, and
		\item[(ii)] $(\psi_1\land\psi_2)\in \lambda(i)$ if and only if $\{\psi_1,\psi_2\}\subseteq \lambda(i)$, for all $(\psi_1\land\psi_2)\in \Sigma$, 
	\end{itemize}
	\item \label{tab:counting} For each $i<t$ and $(\exists[\leq c]x \ \psi)\in \Sigma$, we require that
	\begin{equation*}
		(\exists[\leq c]x \ \psi) \in \lambda(i) \qquad \iff \qquad \card{\{j< t : \psi\in \lambda(j)\}}\ \leq\ c
	\end{equation*}
	\item \label{tab:successor} If $k>0$, then for each $i<t$ and $\D\psi\in \lambda(i)$, there exists some $t'\in \{t,\dots,N\}$ and a labelling function $\lambda':\indexset{t'}\to \Pow{\Sigma}$, such that:
	\begin{itemize}
		\item[(i)] $\psi\in \lambda'(i)$,
		\item[(ii)] $\xi\in \lambda'(j)$ only if $\D\xi\in \lambda(j)$, for all $\D\xi\in\Sigma$ and $j\leq t$,
		\item[(iii)] And $\QKworld[\Sigma,N](\psi;k-1,t',\lambda')$ returns \true.
	\end{itemize}
\end{enumerate}	
\medskip

As with the $\Kworld$ function as described by Spaan~\cite{Spaan1993}, the above function is non-deterministic as it must explore all possible choices for the parameters $t'$ and $\lambda'$ demanded by condition \ref{tab:successor}. However, by appealing to Savitch's theorem~\cite{Savitch1970}, we require only that the algorithm uses at most a polynomial amount of space along any possible computation.

Each of the checks for \ref{tab:boolean} can be performed `in-place' on the graph of $\lambda$. The only part that requires further scrutiny is the additional check on the consistency of the counting quantifiers in \ref{tab:counting}, which can be achieved by with the aid of a counter that is incremented for every $\psi\in \lambda(j)$ located from among those $j< t$ and then compared against whether $(\exists[\leq c] \psi)\in \lambda(i)$, for each $i<t$. The additional space requirements for such a counter is at most linear in the size of $t$ (or even logarithmic if the counter is encoded in binary). 
If $k=0$ then there is nothing to check for \ref{tab:successor}. Otherwise, each of the checks \ref{tab:successor}(i)--\ref{tab:successor}(ii) can also be performed `in-place' on the graphs of $\lambda$ and $\lambda'$, both of which are at most polynomial in the size of $N$ and $\card{\Sigma}$. Finally, the depth of the recursion required for \ref{tab:successor}(iii) is bounded by $k$, and so it follows that \ref{tab:successor} can be perfomed in polynomial space.

%

For soundness and correctness, we claim that $\QKworld[\Sigma,N](k,t,\lambda)$ returns \true\ if and only if there is some expanding domain model $\M=(\F,D,\dom,\I)$, with domain $D=\indexset{N}$, such that:
\begin{enumerate}[label={\bf (I.H.\arabic*)}]\itemindent=1cm
	\item \label{IH:tab1} $\F$ is an irreflexive, intransitive tree of depth $\leq k$,
	\item \label{IH:tab2} $\card{\dom(r)}= t$, where $r$ is the root of $\F$, 
	\item \label{IH:tab3} $M,r\models^{j} \psi$ if and only if $\psi\in \lambda(j)$, for all $j<t$ and $\psi\in \Sigma$.
\end{enumerate} 

\smallskip

We prove this by induction on $k$, so let $k\geq 0$ be fixed, and if $k>0$ then suppose that the claim holds for all $m<k$:

\begin{itemize}\small
	\lefttoright Suppose that $\QKworld[\Sigma,N](k,t,\lambda)$ returns \true. If $k>0$ then, by \ref{tab:successor}, for each $i<t$ and $\D\psi\in \lambda(i)$, there must be some $t\in \bbN$ such that $t\leq t'\leq N$ and $\lambda_{i,\psi}:\indexset{t_{i,\psi}} \to \Pow{\Sigma}$ such that $\psi\in \lambda_{i,\psi}(i)$ and $\QKworld[\Sigma,N](k-1,t_{i,\psi},\lambda_{i,\psi})$ returns \true. 
	By the induction hypothesis, there is some expanding model $\M_{i,\psi}=(\F_{i,\psi},D_{i,\psi},\dom_{i,\psi},\I_{i,\psi})$ satisfying conditions \eqref{IH:tab1}--\eqref{IH:tab3}. 
	From this (possibly empty) collection of models, we define a new expanding domain model $\M=(\F,D,\dom,\I)$, by taking $\F$ to be the result of connecting each of the trees $\F_{i,\psi}$ together at a common root $r$. 
	We define $\dom$ by taking $\dom(w)=\dom_{i,\psi}(w)$, whenever $w$ belongs to $\F_{i,\psi}$, and setting $\dom(r) = \indexset{t}$, so that $\card{\dom(r)} = t$. Similarly, let $\I(w,P_j)= \I_{i,\psi}(w,P_j)$ whenever $w$ belongs to $\F_{i,\psi}$ and let
	\begin{equation*}
		\I(r,P_j) \ = \ \{j<t : P_j(x) \in \lambda(j)\},
	\end{equation*}
	for all $P_i\in \pred$ occurring in $\phi$. 
	
It follows from this construction that $\F$ is an irreflexive, intranstive tree of depth $\leq k$, as required for \ref{IH:tab1}, while $\card{\dom(r)}=t$, as required for \ref{IH:tab2}. 
Finally, it follows from a routine induction 
that $\M,r\models^i \xi$ if and only if $\xi\in \lambda(i)$, as required for \ref{IH:tab3}.

\begin{draftproof}
\begin{itemize}
	\item For \eqref{tab:boolean}(i)
	\begin{eqnarray*}
		\neg \psi \in \lambda(i) & \iff & \M,r\models^i \neg \psi\\
		& \iff & \M,r\not\models^i \psi\\
		& \iff & \psi\not\in \lambda(i)
	\end{eqnarray*}

	\item For \eqref{tab:boolean}(ii)
	\begin{eqnarray*}
		(\psi_1\land \psi_2)\in \lambda(i) & \iff & \M,r\models^i (\psi_1\land \psi_2)\\
		& \iff & \M,r\models^i \psi_1 \mbox{ and } \M,r\models^i\psi_2\\
		& \iff & \psi_1\in \lambda(i) \mbox{ and } \psi_2\in \lambda(i)
	\end{eqnarray*}
	
	\item For \eqref{tab:counting}, we have that
		\begin{eqnarray*}
			(\exists[\leq c] \psi) \in \lambda (i) & \iff & \M,r\models^i (\exists[\leq c]x \psi)\\
			& \iff & \card{\{j \in \dom(r) : \M,r\models^j \psi\}} \leq c\\
			& \iff & \card{\{j < t : \psi \in \lambda(j)\}} \leq c
		\end{eqnarray*}
\end{itemize}	
\end{draftproof}

\righttoleft Conversely, suppose that there is a expanding domain model $\M=(\F,D,\dom,\I)$ satisfying \eqref{IH:tab1}--\eqref{IH:tab3}. It is a routine exercise to show that each $\lambda(i)$ is Boolean saturated, as required for \ref{tab:boolean}, and that $(\exists[\leq c]x \psi)\in \lambda(i)$ if and only if $\card{\{j<t : \psi\in \lambda(j)\}}\leq c$, for all $i<t$, as required for \ref{tab:counting}.

If $k=0$ then \ref{tab:successor} holds vacuously. Otherwise, suppose that $i<t$ and $\D\psi\in \lambda(i)$. By \eqref{IH:tab3}, we must have that $\M,r\models^i\D\psi$, which is to say that there is some $w\in W$ such that $rR w$ and $\M,w\models^i \psi$. 
Choose $t'=\card{\dom(w)}$ and let $\eta:\indexset{t'}\to \dom(w)$ be an enumeration of $\dom(w)$ such that $\eta(i)=i$, for all $i<t$. We may then choose $\lambda':\indexset{t'}\to \Pow{\Sigma}$ by taking
\begin{equation*}
	\lambda'(j) = \{\xi\in \Sigma : \M,w\models^{\eta(j)} \xi\},
\end{equation*} 
for all $j<t'$. In particular, we have that $\psi\in \lambda'(i)$, as required to \ref{tab:successor}(i). 
Furthermore, if $\xi\in \lambda'(j)$, for some $j<t$, then by definition we have that $\M,w\models^j\xi$, from which it follows that $\M,r\models^j\D\xi$, since $j\in \dom(r)$. Consequently, it follows from \ref{IH:tab3} that that $\D\psi\in \lambda(j)$, as required for \ref{tab:successor}(ii). 
Finally, for \ref{tab:successor}(iii), we note that the submodel $\M_{i,\psi}=(\F',D,\dom',\I')$, where $\F_{i,\psi}$ is the irreflexive, intransitive subtree of $\F$ generated by $w$, satisfies conditions (i)--(iii). Whence, by the induction hypothesis, we have that $\QKworld[\Sigma,N](k-1,t',\lambda')$ returns \true, thereby satisfying \ref{tab:successor}. 
		Hence, we conclude that $\QKworld[\Sigma,N](k,t,\lambda)$ meets all the condition \ref{tab:boolean}--\ref{tab:successor}, and therefore returns \true, as required.	
\end{itemize}

By Lemma~\ref{lem:polysize}, we have that $\phi$ is satisfiable with respect to $\QKexp$ if and only if it can be satisfied in an expanding domain model based on an irreflexive, intransitive tree of depth $\leq m=\md{\phi}$, whose domains do not exceed $N=nm(C+1)$. Hence it follows that $\phi$ is satisfiable with respect to $\QKexp$ if and only if there exists some $t<N$ and some $\lambda:\indexset{t}\to \Pow{\Sigma}$ such that $\QKworld[\Sigma,N](m,t,\lambda)$ returns \true, where $\Sigma=\sub{\phi}$, thereby completing the proof.
\end{proof}

It is noteworthy that despite the lofty \NExpTime-completness of the satisfiability problem for the one-variable counting-free fragment \mbox{$\QK\cap\QMLonecf$} over constant domains, each of the fragments $\QKexp\cap \QMLone[\ell]$ over expanding domains shares the same computational complexity as the underlying propositional modal logic $\K$, for $\ell<\omega$~\cite{Ladner1977}. 
However, despite this, the countable union of each of these \PSpace-complete fragments $\QKexp\cap\QMLone[\ell]$, for $\ell<\omega$, results in the full one-variable fragment $\QKexp\cap \QMLone$, whose satisfiability problem is, once again, \NExpTime-complete.

Note that despite the existence of tableau algorithms for $\QKexp\cap\QMLcf$ in the current literature~\cite{FittingMendelsohn1998}, the question as to the computational complexity of any of its decidable fragments (and of the one-variable fragment, in particular) appears to have been, hitherto, unexamined. 

Consequently, Theorem~\ref{thm:QKexp-bounded} also offers the following new result for the one-variable, counting-free fragment of $\QKexp$.

\begin{corollary}
	The satisfiability problem for the one-variable, counting-free fragment \mbox{$\QKexp\cap \QMLonecf$} is \PSpace-complete.
\end{corollary}

\section{Logics with bounded quantifiers over decreasing domains}
\label{sec:decreasing}

Unlike for logics over expanding domains, we cannot escape the possibility that our decreasing models may require exponentially large domains; that is to say, they do not possess the poly-size domain property. This is true even if we restrict ourselves to the one-variable counting-free fragment $\QMLonecf$, as evidenced by the following formula, adapted from~\cite{Marx1999}:
\begin{multline*}
	\theta_n := \bigwedge_{k=0}^n \forall x \B^k \Big((\D\exists x P_k(x) \land \D\exists x \neg P_k(x)) \\ \land\bigwedge_{\ell<k} (P_\ell(x) \to \B\forall x P_\ell(x)) \land (\neg P_\ell(x) \to \B\forall x \neg P_\ell(x))\Big). 
\end{multline*}
Note that each $\theta_n$, for $n<\omega$, can only be satisfied in models in which the domain at the root node is exponential in the size of $\theta_n$. 
As a result, we cannot emulate the proof of Theorem~\ref{thm:QKexp-bounded} over decreasing domain models to provide us with a \PSpace\ upper-bound on the complexity of $\QKdec\cap\QMLone[\ell]$, for $\ell<\omega$. 
Indeed, we may show that over decreasing domains, the satisfiability problem for the one-variable fragment is \ExpTime-hard, even if we restrict ourself to using only the counting-free quantifiers $\exists x$ and $\forall x$. 
Despite the lack of counting quantifiers, the comutational complexity of this fragments appears to have not yet been uncovered inthe current literature.

\begin{theorem}
\label{thm:QKdec-bounded}
The satisfiability problem for the one-variable counting-free fragment $\QKdec\cap \QMLone[0]$ is \ExpTime-hard.
\end{theorem}
\begin{proof}
The proof is via a reduction from the \ExpTime-complete satisfiability problem for the propositional bimodal logic $\Ku$, characterised by the class of all bimodal Kripke frames $\F=(W,R_1,R_2)$, where $R_2=W\times W$ is the universal relation on $W$~\cite{Spaan1993} (or \cite[Theorem 1.27]{GKWZ03}).

To this end, let $\phi\in \MLu$ an arbitrary bimodal propositional formula in the modal language having a universal modality. Let $S,T\in \pred$ be unary predicate symbols, and for each propositional variable $p_k\in \sub{\phi}$ we associate a fresh unary predicate symbol $P_k\in \pred$. In addition to these, we reserve a fresh auxilliary predicate symbol $Q_\psi\in \pred$, for each subformula $\psi\in \sub{\phi}$. We then define the translation $(\cdot)^\dagger:\sub{\phi}\to \QMLonecf$, by taking
\begin{gather*}
	p_k^\dagger = P_k(x), \qquad 
	(\neg \psi)^\dagger = \neg \psi^\dagger, \qquad
	(\psi_1\land \psi_2)^\dagger = \psi_1^\dagger \land \psi_2^\dagger,\\
	(\D\psi)^\dagger = \D\big(S(x) \land \exists x (T(x) \land  Q_\psi(x))\big),\qquad
	(\D_u\psi)^\dagger = \exists x\ \psi^\dagger.
\end{gather*}
Furthermore, take $\zeta$ to be the conjunction of the following formulas:
\begin{align}
\label{eq:exptime1}
	& \bigwedge_{\psi\in \sub{\phi}}\forall x \big(\psi^\dagger \to \B Q_\psi(x)\big)\\[5pt]
	\label{eq:exptime2}
	& \bigwedge_{\psi\in \sub{\phi}}\forall x \big(\D Q_\psi(x)\to \psi^\dagger\big).
\end{align}

We claim that $\phi$ is satisfiable with respect to $\Ku$ if and only if $(\zeta\land \phi^\dagger)$ is satisfiable with respect to $\QKdec$.

\begin{itemize}
	\righttoleft Suppose that $(\zeta\land \phi^\dagger)$ is satisfiable with repsect to $\QKdec$, which is to say that $\M,r\models^{a_0} \zeta$ and $\M,r\models^{a_0} \phi^\dagger$, for some first-order decreasing model $\M=(\F,D,\dom,\I)$, where $\F$ is an irreflexive intransitive tree.
	
We define a new Kripke frame $\F'=(W',R')$ by taking $W'=\dom(r)$ and 
	\begin{equation*}
		a R' b \;\iff \;\exists u \in W \big(r R u\ \mbox{and}\ a\in \I(u,S) \ \mbox{and}\ b\in \I(u,T)\big)
	\end{equation*}	
	for all $a,b\in W'$. 
	We define a propositional valuation $\V$ by taking $\V(p_i)=\I(r,P_i)$, for all $p_i\in \sub{\phi}$.
	
	\smallskip	
	\noindent
	We prove by induction on the size of $\psi\in\sub{\phi}$ that
	\begin{equation}
	\tag{I.H.3}
	\label{eq:IH1}
		\M',a\models \psi \qquad \iff \qquad \M,r\models^{a} \psi^\dagger
	\end{equation}
	for all $a\in W'$. 
The cases where $\psi$ is an atomic formula or a Boolean combination of smaller formulas are straightforward and follow from the definitions. So suppose that $\psi$ is of the form $\D\xi$ or $\exists[\leq c]x \xi$, for some $\xi\in \sub{\phi}$ and $c\in \bbN$. In which case we have the following:
	
	\begin{itemize}\small
	\begin{draftproof}
		\item \emph{Case $\psi=p_i$}: \; This follows immediately from the definition of $\I(w,P_i)$ since 
		\begin{equation*}
		\M',a\models p_i \;\;
			\iff \;\; a\in \V(p_i)\;\;
			\iff[(def)] \;\; a \in \I(r,P_i)\;\;
			\iff \;\; \M,r\models^{a} P_i(x).
		\end{equation*}
		
		\item \emph{Case $\psi=\neg\xi$}: \; We have that 
		\begin{eqnarray*}
			\M',a\models \neg\xi
				& \iff & \M',a\not\models \xi\\
				& \iff[\eqref{eq:IH1}] & \M',r\not\models^{a} \xi^\dagger\\
				& \iff & \M,r\models^{a} \neg\xi^\dagger\\
				& \iff & \M,r\models^{a} (\neg\xi)^\dagger.
		\end{eqnarray*}		
		
		\item \emph{Case $\psi=(\xi_1 \land \xi_2)$}: \; We have that 
		\begin{eqnarray*}
			\M',a\models (\xi_1 \land \xi_2)
				& \iff & \M',a\models \xi_1 \mbox{ and } \M',a\models \xi_2\\
				& \iff[\eqref{eq:IH1}] & \M',r\models^{a} \xi_1^\dagger \mbox{ and } \M',r\models^{a} \xi_2^\dagger\\
				& \iff & \M,r\models^{a} (\xi_1^\dagger \land \xi_2^\dagger)\\
				& \iff & \M,r\models^{a} (\xi_1 \land \xi_2)^\dagger.
		\end{eqnarray*}		
		\end{draftproof}

		\item \emph{Case $\psi=\D_u\xi$}: \; We have that 
		\begin{eqnarray*}
			\M',a\models \D_u\xi
				& \iff & \exists b\in \dom(r);\ \M',b\models \xi\\
				& \iff[\eqref{eq:IH1}] & \exists b\in \dom(r);\ \M',r\models^{b} \xi^\dagger\\
				& \iff & \M,r\models^{a} \exists x\ P_k(x).
		\end{eqnarray*}
		
		\item \emph{Case $\psi=\D\xi$}: \; 
		If $\M',a\models \D\xi$ then there is some $b\in W'$ such that $a R'b$ and $\M',b\models \xi$. By the induction hypothesis, we have that $\M,r\models^{b} \xi^\dagger$. Whence, by \eqref{eq:exptime1}, we have that $\M,r\models^{b} \B Q_\xi(x)$. 
		By the definition of $R'$, there is some $u\in W$ such that $r R u$, $a\in \I(u,S)$ and $b\in \I(u,T)$. It follows that $\M,u\models^b T(x) \land Q_\xi(x)$ and so $\M,u\models^a S(x) \land \exists x (T(x) \land Q_\xi(x))$. Therefore, $\M,r\models^a \D(S(x) \land \exists x(T(x) \land Q_\xi(x)))$, which is to say that $\M,r\models^a (\D\xi)^\dagger$, as required.
		
		Conversely, suppose that $\M,r\models^a \D(S(x) \land \exists x (T(x) \land  Q_\xi(x)))$. Then there is some $u\in W$ such that $\M,u\models^a S(x) \land \exists x (T(x) \land Q_\xi(x))$. It follows that $a\in \I(u,S)$ and there is some $b\in \I(u,T)\subseteq \dom(u)$ such that $\M,u\models^b Q_\xi(x)$. Note that by the definition of $R'$, we have that $a R' b$. Furthermore, since $\M$ is decreasing, we have that $b\in \dom(r)\supseteq \dom(u)$, and so $\M,r\models^b \D Q_\xi(x)$. Whence, by \eqref{eq:exptime2}, we have that $\M,r\models^b \xi^\dagger$, and so it follows from the induction hypothesis that $\M',b\models \xi$. We then have that $\M',a\models \D \xi$, as required. 
	
	\end{itemize}
	
	Hence it follows that $\M',a\models \psi$ if and only if $\M,r\models^{a} \psi^\dagger$, for all $\psi\in \sub{\phi}$ and $a\in W'$. In particular, we have that $\M',a_0\not\models \phi$, which is to say that $\phi$ is satisfiable with respect to $\Ku$, as required.

	\righttoleft Suppose that $\phi$ is satisfiable with respect to $\Ku$. Then $\M,r\not\models \phi$ for some propositional Kripke model $\M=(\F_u,\V)$, where $\F=(W,R)$ is an irreflexive, intranstive tree and $\F_u=(W,R,W\times W)$. 
	
	\newcommand{\rootnode}{\circ}
	We construct a new frame $\F'=(W',R')$ by taking
	\begin{equation*}
		W' = W\cup \{\rootnode\}, \qquad \mbox{and} \qquad R' = \{(\rootnode,w) : w\in W\}
	\end{equation*}
	where $\rootnode\not\in W$ is a newly introduced root node, from which all other worlds are accessible. 
	We construct a first-order Kripke model on $\F'$ by taking $\M'=(\F',D',\dom,\I)$, where  $D' = W$, $\dom(\rootnode)=\dom(w) = D'$, for all $w\in W$, and taking:
	\begin{gather*}
	\I(u,S) = \{u\}, \qquad \I(u,T) = \{v \in W : u R v\}, \qquad \I(u,P_i)=\emptyset,\\
	\I(u,Q_\xi) = \{v\in W: \M,v\models \xi\}
	\end{gather*}
	for all $u\in W$, with
	\begin{gather*}
	\I(\rootnode,S) = \emptyset, \qquad \I(\rootnode,T) = \emptyset, \qquad \I(\rootnode,P_i)=\V(p_i),\qquad
	\I(\rootnode,Q_\xi) = \emptyset
	\end{gather*}
for all $p_i\in \sub{\phi}$.
	
\smallskip
\noindent
We prove by induction on the size of $\psi\in\sub{\phi}$ that
	\begin{equation}
	\tag{I.H.4}
	\label{IH:QKdec-RtoL}
		\M',\rootnode\models^{w} \psi^\dagger \qquad \iff \qquad \M,w\models \psi
	\end{equation}
	for all $w\in W$. 
	Again, the cases where $\psi$ is an atomic formula or a Boolean combination of smaller formulas are straightforward and follow from the definitions. So suppose that $\psi$ is of the form $\D\xi$ or $\exists[\leq c]x \xi$, for some $\xi\in \sub{\phi}$ and $c\in \bbN$. We then have the following cases:
	
	\begin{itemize}\small
	\begin{draftproof}
		\item \emph{Case $\psi=p_i$}: \; This follows immediately from the definition of $\I(\rootnode,P_i)$, since 
		\begin{equation*}
		\M',\rootnode\models^{w} P_i(x) 
				\;\; \iff \;\; w\in \I(\rootnode,P_i)
				\;\; \iff[(def)] \;\; w \in \V(p_i)
				\;\; \iff \;\; \M,w\models p_i.
		\end{equation*}
		
		\item \emph{Case $\psi=\neg\xi$}: \; We have that
		\begin{eqnarray*}
			\M',\rootnode\models^{w} (\neg\xi)^\dagger 
				& \iff & \M',\rootnode\models^{w} \neg\xi^\dagger\\
				& \iff & \M',\rootnode\not\models^{w}\xi^\dagger\\
				& \iff[\eqref{IH:QKdec-RtoL}] & \M,w\not\models \xi\\
				& \iff & \M,w\models \neg \xi
		\end{eqnarray*}
		
		\item \emph{Case $\psi=(\xi_1\land \xi_2)$}: \; We have that
		\begin{eqnarray*}
			\M',\rootnode\models^{w} (\xi_1 \land \xi_2)^\dagger 
				& \iff & \M',\rootnode\models^{w} \xi_1^\dagger \land \xi_2^\dagger\\
				& \iff & \M',\rootnode\models^{w}\xi_1^\dagger \mbox{ and } \M',\rootnode\models^{w}\xi_2^\dagger\\
				& \iff[\eqref{IH:QKdec-RtoL}] & \M,w\models \xi_1 \mbox{ and } \M,w\not\models \xi_2\\
				& \iff & \M,w\models \xi_1 \land \xi_2
		\end{eqnarray*}
		
		\item \emph{Case $\psi=\D_u\xi$}: \; We have that
		\begin{eqnarray*}
			\M',r\models^{w} (\D_u\xi)^\dagger 
				& \iff & \M',r\models^{w} \exists x\ \xi^\dagger\\
				& \iff & \exists u\in W; \; \M',r\models^{u} \xi^\dagger\\
				& \iff[\eqref{IH:QKdec-RtoL}] & \exists u\in W;\ \M,u\models\xi\\
				& \iff & \M,w\models \D_u \xi
		\end{eqnarray*}
		
		\end{draftproof}
		\item \emph{Case $\psi=\D\xi$}: \; 
	If $\M',\rootnode \models^{w} \D(S(x) \land \exists x (T(x) \land  Q_\xi(x)))$, then there is some $u\in W$ such that $\M,u\models^{w} S(x) \land \exists x(T(x) \land Q_\xi(x))$. 
	It follows from the definition of $\I(u,S)$ that $u=w$. Furthermore, there is some $v\in W$ such that $\M,u\models^{v} T(x) \land Q_\xi(x)$. 
	By the definitions of $\I(u,T)$ and $\I(u,Q_\xi)$, we have that $u R v$ and $\M,v\models \xi$. Hence $\M,u\models \D\xi$, which is to say that $\M,w\models \D\xi$, since $u=w$, as required.
	
	Conversely, suppose that $\M,w\models \D\xi$. Then there is some $v\in W$ such that $w R v$ and $\M,v\models \xi$. It follows from the definition that $w\in \I(w,S)$, $v \in \I(w,Q_\xi)$ and $v\in \I(w,T)$. We then have that $\M',w\models^{w} S(x) \land \exists x (T(x) \land Q_\xi(x))$. Furthermore, since $\rootnode R' w$, for all $w\in W$, we have that $\M',\rootnode\models^w \D(S(x) \land \exists x (T(x) \land Q_\xi(x)))$, which is to say that $\M',\rootnode\models^w (\D\xi)^\dagger$, as required.
	\end{itemize}
	
Hence it follows that $\M',\rootnode\models^{w} \psi^\dagger$ if and only if $\M,w\models \psi$, for all $\psi\in \sub{\phi}$ and $w\in W$. In particular, we have that $\M',\rootnode\models^{r} \phi^\dagger$. 
Therefore, it remains to show that $\M',\rootnode\models^{r} \zeta$.
\begin{itemize}
	\item For \eqref{eq:exptime2}, suppose that $w\in W$ is such that $\M',\rootnode\models^{w} \D Q_\psi(x)$. Then there is some $u\in W'$ such that $\rootnode R' u$ and $\M',u\models^{w} Q_\psi(x)$. By definition we have that $\M,w\models \psi$. It then follows from \eqref{IH:QKdec-RtoL} that $\M',\rootnode\models^{w} \psi^\dagger$. Consequently, we have that $\M'\rootnode\models^{a_0}\forall x (\D Q_\psi(x) \to \psi^\dagger)$.
	
	\item For \eqref{eq:exptime1}, suppose that $w\in W$ is such that $\M',\rootnode\models^{w} \psi^\dagger$. Then by \eqref{IH:QKdec-RtoL}, we have that $\M,w\models \psi$. By definition, $w\in \I(u,Q_\psi)$, for all $u\in W$ and so $M',\rootnode\models^{w} \B Q_\psi(x)$. Consequently, we have that $\M',\rootnode \models^{a_0} \forall x (\psi^\dagger \to \B Q_\psi(x))$.
\end{itemize}	
	
Hence, we have that $\M',\rootnode\models^{r} (\zeta\land \phi^\dagger)$, which is to say that $(\zeta\land \phi^\dagger)$ is satisfiable with respect to $\QKdec$, as required.
\end{itemize}

Since the satisfiability problem for the propositional modal logic $\Ku$ is \ExpTime-hard, so too must be that of the one-variable counting-free fragment $\QKdec\cap \QMLone[0]$, as required.
\end{proof}

This result, together with that of Theorem~\ref{thm:QK-unbounded}, places the complexity of the satisfiability problem for each of the fragments $\QKdec\cap\QMLone[\ell]$, for $\ell<\omega$, and that of $\QKdec\cap \QMLoneun$ between \ExpTime-hard and \NExpTime. However, it remains open as to where their precise complexities lie.

\begin{question}
Is the complexity of the satisfiability problem for each of the fragments $\QKdec\cap \QMLone[\ell]$, for $\ell<\omega$, strictly less than that of their union $\QKdec\cap\QMLone$?
\end{question}
\begin{question}
	Is the complexity of the satisfiability problem for $\QKdec\cap\QMLone[0]$ strictly less than that of $\QKdec\cap\QMLone[1]$? 
\end{question}

\section{Applications to two-dimensional propositional modal logics}
\label{sec:bimodal}
First-order modal logics are intimately related to another extensively studied formalism; that of many-dimensional modal logics~\cite{Shehtman1978,Segerberg1973,GKWZ03,Kurucz2007,MarxVenema1997}.  %
Given a countably infinite set of propositional variables $\propvar=\{p_0,p_1,\dots\}$, let $\ML_2$ denote the set of bimodal formulas defined in accordance to the following grammar:
\begin{equation*}
\phi\ ::=\ p_i \ \mid \ \neg \phi \ \mid \ (\phi_1\land \phi_2) \ \mid \ \Dh\phi \ \mid \ \Dv \phi
\end{equation*}
where $p_i\in \propvar$, and $\Dh$ and $\Dv$ are modal operators, with subscripts suggestive of the `horizontal' and `vertical' dimensions in which they are to operate. 
Formulas of $\ML_2$ are interpreted over Kripke models $\M=(\F,\V)$, where $\F=(W,R_h,R_v)$ is a \emph{bimodal Kripke frame}, with $R_h,R_v\subseteq W^2$, and \mbox{$\V:\propvar \to \Pow{W}$} is a propositional valuation on $\F$. Satisfiability is defined in the usual way with $\D_j\phi$ being interpreted by the relation $R_j$, for $j=h,v$. 
Of particular interest are \emph{product} models in which the two modal operators act orthogonally: We define the product of two unimodal frames $\F_h=(W_h,R_h)$ and $\F_v=(W_v,R_v)$ to be the bimodal frame $\F_h\times \F_v= (W_h\times W_v,\overline{R}_h,\overline{R}_v)$, where
\begin{eqnarray*}
	(u,v) \overline{R}_h (u',v') &\quad \iff \quad &  u R_h u'\ \mbox{and}\ v=v',\\[5pt]
	(u,v) \overline{R}_v (u',v') & \iff &  u=u'\ \mbox{and}\ v R_v v',
\end{eqnarray*}
for all $u,u'\in W_h$ and $v,v'\in W_v$. 
A formula $\phi\in \ML_2$ is said to be \emph{satisfiable with respect to $L_h\times L_v$} if it is satisfiable in some product model $\F_h\times \F_v$, where $\F_h$ and $\F_v$ are frames for $L_h$ and $L_v$, respectively.

%
%
The product construction was first described by Segerberg in~\cite{Segerberg1973} for the case where both components were frames for $\Sfive$, and was later generalised to arbitrary frames by Shehtman~\cite{Shehtman1978}.

A natural extension of the product construction is to consider \emph{subframes} of product frames, in which only a subset of the possible worlds of $\F_h\times \F_v$ are possible. 
We say that $\G=(W,R_h,R_v)$~is:
\begin{itemize}
	\item[--] an \emph{expanding product model} if $(u,v)\in W$ implies $(u,v')\in W$, and 
	\item[--] a \emph{decreasing product model} if $(u,v')\in W$ implies $(u,v)\in W$,
\end{itemize}
whenever $v R_h v'$.

It was observed by Wasjberg~\cite{Wajsberg1933} that the one-variable fragment of first-order logic can be interpreted as a syntactic variant of the modal logic $\Sfive$ of all equivalence frames. This can be naturally extended to the one-variable fragment of first-order modal logics under the translation that maps $\D\phi := \Dh\phi$ and $(\exists x \phi) := \Dv\phi$~\cite{GabbayShehtman1993}. 
From this, we obtain the following proposition. 

\newcommand{\subf}{\mathit{sf}}
\begin{proposition}
\label{prop:LxS5}
	{\rm (i)} The satisfiablity problem for $\bfL\times \Sfive$ is equivalent to that of $\QL\cap \QMLonecf$. {\rm (ii)} The satisfiablity problem for $\bfL\times^{\subf} \Sfive$ is equivalent to that of $\QL^{\subf}\cap \QMLonecf$, for $\subf\in \{\ex,\dc\}$.
\end{proposition}

A closely related sub-logic of $\Sfive$ is von~Wright's `logic of elsewhere' $\Diff$, characterized by the class of all \emph{difference frames} of the form $(W,R_{\ne})$ in which $uR_{\ne} v$ if and only of $u \not=v$, for all $u,v\in W$~\cite{vonWright1979}. Segerberg later provided a complete axiomatisation for $\Diff$, identifying it as the logic of all symmetric, weakly-transitive frames, with the following axioms~\cite{Segerberg1980}:
\begin{equation*}
\mathit{(sym)} = p \to \B\D p\quad \mbox{and} \quad
\mathit{(wtran)} = \D\D p \to \D p \lor p
\end{equation*}

By extending Wajsberg's result for $\Sfive$, we can reduce the satisfiability problem for $\bfL\times \Diff$ to that of with a fragment of $\QL\cap \QMLone[1]$, with the aid of some counting quantifiers.

\begin{proposition}
\label{prop:LxDiff}
{\rm (i)} The satisfiablity problem for $\bfL\times \Diff$ is equivalent to that of $\QL\cap \QMLone[1]$. {\rm (ii)} The satisfiablity problem for $\bfL\times^{\subf} \Diff$ is equivalent to that of $\QL^{\subf}\cap \QMLone[1]$, for $\subf\in \{\ex,\dc\}$.
\end{proposition}
\begin{proof}
Let $\phi\in \ML_2$ be a propositional bimodal formula, and let $P_i\in \pred$ be a unary predicate symbol associated with each propositional variable $p_i\in \sub{\phi}$. 
We define the translation $(\cdot)^\dagger:\ML_2 \to \QMLone[1]$ by taking
\begin{gather*}
	p_i^\dagger = P_i(x), \qquad (\neg \psi)^\dagger = \neg \psi^\dagger, \qquad (\psi_1\land \psi_2)^\dagger = \psi_1^\dagger \land \psi_2^\dagger,\\[5pt]
	(\Dh\psi)^\dagger = \D\psi^\dagger, \qquad 
	(\Dv\psi)^\dagger = Q_\psi(x),
\end{gather*}
for each $p_i\in \sub{\phi}$. 
Take $\zeta\in \QMLone[1]$ to be the following conjunction:
\begin{equation}
	\bigwedge_{\psi\in \sub{\phi}}\forall x \big(Q_{\psi}(x) \liff \exists^{\ne}x\ \psi\big),
\end{equation}
where $\exists^{\ne}x\ \psi := (\neg \psi^\dagger \land \neg \exists[\leq 0] x\ \psi^\dagger) \lor \neg \exists[\leq 1] x\ \psi^\dagger$ specifies the existence of some \emph{other} domain object satisfying $\psi$. Note that, since $\size{\psi^\dagger}\leq \size{\psi}$, for all $\psi\in \sub{\phi}$, we have that $\zeta$ is at most polynomial in the size of $\phi$.

For each subframe product model $\M=(\G,\V)$, where $\G\subseteq \F_h\times \F_v$, we associate a first-order Kripke model $\M^\star=(\F,D,\dom,\I)$, by taking $\F=(W_h,R_h)$, $D=W_v$, $\dom(u) = \{v \in W_v : (u,v)\in V\}$ and 
\begin{eqnarray*}
	\I(u,P_i) & = & \{v\in W_v : (u,v)\in \V(p_i)\},\\
	\I(u,Q_\psi) & = & \{v\in W_v : \M,(u,v)\models\D \psi\},
\end{eqnarray*} 
for all $u\in W_h$, $p_i\in \propvar$ and $\psi\in \sub{\phi}$.
It then follows from a routine induction that $\phi$ is satisfiable in $\M$ if and only if $(\zeta \land \phi^\dagger)$ is satisfiable in~$\M^\star$. 
Furthermore, it is straightforward to check that $\M^\star$ is an expanding (resp. decreasing) domain model precisely when $\M$ is an expanding (resp. decreasing) product model, therby completing the proof.
\end{proof}

By taking $\bfL$ to be the minimal modal logic $\K$, Propositions~\ref{prop:LxS5} and \ref{prop:LxDiff} yield the following corollaries of Theorems~\ref{thm:QK-unbounded},\ref{thm:QKexp-bounded}, and \ref{thm:QKdec-bounded}:

\begin{corollary}The satisfiability problem for:
\begin{itemize}\itemindent=1em
	\item[\textup{(i)}] $\K\times\Diff$ is  \mbox{\NExpTime-complete}\footnote{Is is already well-established that $\K\times\Sfive$ is also  \NExpTime-complete~\cite{Marx1999}.},
	\item[\textup{(ii)}] $\K\etimes \bfL$ is \PSpace-complete, for $\bfL\in \{\Sfive,\Diff\}$,
	\item[\textup{(iii)}] $\K\dtimes \bfL$ is \ExpTime-hard in \NExpTime, for $\bfL\in \{\Sfive,\Diff\}$.
\end{itemize}
\end{corollary}

These results mark a stark contrast against the negative results one often faces when taking two-dimenisonal products with von~Wright's logic, which are often vastly more complex than their corresponding $\Sfive$-counterparts. In particular, the satisfiability problem for $\Ku\times\Diff$ is undecidable~\cite{Hampson2012}, while that of $\Ku\times\Sfive$ is decidable in \NExpTime[2]~\cite[Theorem 6.5]{GKWZ03}. A similar jump in complexity arises for products in which the horizontal component is characterised by some class of linear orders, such as with $\Kfourt\times\Sfive$ whose satisfibility problem is decidable in \ExpTime[2], whereas that of $\Kfourt\times\Diff$ is undecidable~\cite{Hampson2015}.

%

\section{Discussion}
\label{sec:discussion}
Throughout this paper, we have focused our attention on fragments of $\QML$ comprising a single first-order variable, owing to the wealth of negative results that already exist for two-variable modal logics~\cite{Kripke1962}. However, there remains scope to investigate the role of counting quantifiers to fragments that lie beyong the one-variable fragment, such as the \emph{monodic fragment}, briefly described in Section~\ref{sec:intro}~\cite{Wolter1998}

However, it should be noted that there is no immediate application of the techniques developed in~\cite{Wolter1998}, which the authors employed to prove that the (counting-free) monodic fragment of $\QKc$ is decidable, where $\Kc$ denotes the bimodal logic of all frames whose second relation is the transitive closure of the first. It is known, however, that even the one-variable fragment $\QKc\cap \QMLone[1]$, whose sole counting quantifiers are $\exists[\leq 0] x$ and $\exists[\leq 1] x$, is already non-recursively enumerable~\cite{Hampson2012} (indeed, even highly undecidable~\cite[Theorem~8.5]{HampsonPhD}).\\

The one-variable counting-free fragment $\QKfour \cap \QMLonecf$ is known to admit filtration and so it's satsifiability problem can be decided in \NExpTime[2]~\cite{GabbayShehtman1998,GKWZ03}, where $\Kfour$ is the logic of all \emph{transitive frames}. 
However, it remains open whether the the satisfiability problem for the full one-variable fragment \mbox{$\QKfour\cap \QMLone$} or for any of the fragments $\QKfour \cap \QMLone[\ell]$ are decidable, for $0<\ell<\omega$. 
It is tempting to consider whether the results of Section~\ref{sec:expanding}, can be adapted to provide a similar \PSpace\ upper-bound on the satisfaibility problem for $\QKfour^\ex\cap \QMLone[\ell]$, analogous to Ladner's $\mathsf{K4\mbox{-}WORLD}$ algorithm. However, this approach fails since the finite domain property, proved in Lemma~\ref{lem:polysize}, was contingent on every posible world having at most one predecessor. Indeed, it is not difficult to construct examples of satisfiable formulas that cannot be satisfied in models having only finitely many domain elements~\cite[Theorem 5.32]{GKWZ03}.\\

The main question left open in this paper asks what is the precisely complexity of the satisfiability problem for each of the fragments $\QKdec \cap \QMLone[\ell]$, for $\ell<\omega$. 
One might suppose that, given the lack of the poly-size domain property, that one may be able to derive an \NExpTime-hard lower-bound via a reduction from the $(2^n\times 2^n)$-tiling problem~\cite{emdeBoas1997}, \textit{a la} Marx~\cite{Marx1999}. 
However, the reduction employed by Marx relies heavily on both left and right commutativity between the modal operators and first-order quantifiers.

Indeed, part of the problem we face with logics over decreasing domains is the inability for branches to `communicate' directly with one another as they do in constant domain models, where every object in one branch is shared between every other branch of the model. Contrast this with the situation in decreasing domain models, where the first-order domains at each of the leaves of the underlying frame may be wholly disjoint from one~another. 
There may, therefore, be hope that an \ExpTime\ algorithm for this fragment may yet be uncovered.

\appendix
\section{Implementation of QKworld Algorithm}
\label{app:QKworld}
\strut\\
\newcommand{\Continue}{\textbf{next}}
\newcommand{\Break}{\textbf{break}}
\newcommand{\false}{\textbf{FALSE}}
\newcommand{\true}{\textbf{TRUE}}
\newcommand{\Input}[1]{\textbf{Input:} #1\\}
\begin{algorithm}[H]\small
\SetAlgoLined
\Input{$\Sigma$, $N$, $k$, $t$, $\lambda$}
\If{$t>N$}{\Return \false}

\ForAll{$i\leq t$}{
	\ForAll{$\neg\psi\in \Sigma$}{
		\textbf{if} $\neg\psi\in \lambda(i)$ and $\psi\in \lambda(i)$ \textbf{then} \Return \false\;
		\textbf{if} $\neg\psi\not\in \lambda(i)$ and $\psi\not\in \lambda(i)$ \textbf{then} \Return \false\;
	}
	\ForAll{$(\psi_1\land \psi_2)\in \Sigma$}{
		\textbf{if} $(\psi_1\land \psi_2)\in \lambda(i)$ and $\{\psi_1,\psi_2\}\not\subseteq \lambda(i)$ \textbf{then} \Return \false\;
		\textbf{if} $(\psi_1\land \psi_2)\not\in \lambda(i)$ and $\{\psi_1,\psi_2\}\subseteq \lambda(i)$ \textbf{then} \Return \false\;
	}
	
\ForAll{$(\exists[\leq c]x \ \psi)\in \Sigma$}{
	$\mathsf{count} := 0$\;	
	\ForAll{$j\leq t$}{
		\textbf{if} $\psi\in \lambda(j)$ \textbf{then} $\mathsf{count} = \mathsf{count} + 1$\;
	}
	\textbf{if} $(\exists[\leq c]x\ \psi)\in \lambda(i)$ and $\mathsf{count}>c$ \textbf{then} \Return \false\;
	\textbf{if} $(\exists[\leq c]x\ \psi)\not\in \lambda(i)$ and $\mathsf{count}\leq c$ \textbf{then} \Return \false\;
}
}

\ForAll{$i\leq t$ \textup{and} $\D\psi\in \lambda(i)$}{
	$\mathsf{flag} :=\false$\;
	\ForAll{$t'\in\{t,\dots, N\}$ \textup{and} $\lambda':\{0,\dots, t'\} \to \Pow{\Sigma}$}{
		\textbf{if} $\psi\not\in \lambda'(i)$ \textbf{then} \Continue\;
		\textbf{if} \textup{$\QKworld(\Sigma,N,k-1,t',\lambda')=\false$} \textbf{then} \Continue\;
		$\mathsf{flag} := \true$\;
		\Break
	}
	\textbf{if} \textup{$\mathsf{flag}=\false$} \textbf{then} \Return \false\;

}
\Return \true

 \caption{$\QKworld$ algorithm}
\end{algorithm}


\end{document}